\documentclass[a4paper,twoside, 11pt]{journal}
\usepackage{fullpage}

\usepackage{amssymb,amsmath,amsthm,wasysym,mathrsfs}
\usepackage{xspace}
\usepackage{graphicx}
\usepackage[dvipsnames]{xcolor}
\usepackage{tabularx}
\usepackage{booktabs}
\usepackage{wrapfig,refcount}
\usepackage{microtype}
\usepackage{hyperref}

\hypersetup{
  colorlinks=true,
  linkcolor= black,
  citecolor= black,
  pageanchor=true,
  urlcolor = black,
  plainpages = false,
}

\newcommand{\mathfunc}[1]{\ensuremath{\mathit{#1}}\xspace}
\newcommand{\mathapply}[2]{\ensuremath{\mathit{#1}(#2)}\xspace}

\newcommand{\afunc}[1]{\textsc{#1}}
\newcommand{\acall}[2]{\ensuremath{\afunc{#1}(#2)}}

\newcommand{\mkmcal}[1]{\ensuremath{\mathcal{#1}}\xspace}

\newcounter{tmpthm}

\newcommand{\G}{\mkmcal{G}}
\newcommand{\Lst}{\mkmcal{L}}
\newcommand{\T}{\mkmcal{T}}
\newcommand{\C}{\mkmcal{C}}
\newcommand{\X}{\mkmcal{X}}
\newcommand{\D}{\mkmcal{D}}

\newcommand{\E}{\mkmcal{E}}
\newcommand{\A}{\mkmcal{A}}
\renewcommand{\S}{\mkmcal{S}}
\renewcommand{\H}{\mkmcal{H}}

\newcommand{\RG}{\mkmcal{R}}

\renewcommand{\P}{\mkmcal{P}}
\newcommand{\taun}{n\tau}

\newcommand{\mkmbb}[1]{\ensuremath{\mathbb{#1}}\xspace}

\newcommand{\R}{\mkmbb{R}}

\newcommand{\PP}{\mkmbb{P}}

\newcommand{\eps}{\varepsilon}

\newcommand{\etal}{et al.\xspace}

\newcommand{\codots}{,\hspace{-1pt}.\hspace{.5pt}.\hspace{.5pt},}

\newcommand{\duration}{\ensuremath{\mathit{duration}}\xspace}
\newcommand{\Time}{\mkmbb{T}}

\newcommand{\below}[1]{\mathapply{below}{#1}}

\newcommand{\gstart}[1]{\mathapply{start}{#1}}

\DeclareMathOperator{\symd}{\Delta}
\DeclareMathOperator{\polylog}{polylog}

\theoremstyle{plain}
\newtheorem{theorem} {Theorem}
\newtheorem{lemma}[theorem] {Lemma}

\newtheorem{remark} {Remark}
\newtheorem{observation}[theorem]{Observation}

\newcommand{\myparNS}[1]{\noindent{\sffamily\bfseries #1.}}
\newcommand{\mypar}[1]{\medskip\myparNS{#1}}

\title{Grouping Time-varying Data\protect\\ for Interactive Exploration}

\author{
  Arthur van Goethem\thanks{
    Department of Mathematics and Computer Science, Eindhoven University of Technology
    \texttt{\{a.i.v.goethem,b.speckmann\}@tue.nl}
  }
  \and
  Marc van Kreveld\thanks{
    Department of Information and Computing Sciences, Utrecht University,
    \texttt{\{m.j.vankreveld, m.loffler\}@uu.nl}
  }
  \and
  Maarten L\"offler\footnotemark[2]
  \and
  Bettina Speckmann\footnotemark[1]
  \and
  Frank Staals\thanks{
    MADALGO, Aarhus University,
    \texttt{f.staals@cs.au.dk}
  }
}

\begin{document}

\maketitle

\begin{abstract}
%
We present algorithms and data structures that support the interactive
analysis of the grouping structure of one-, two-, or higher-dimensional time-varying data
while varying all defining parameters. Grouping structures characterise important patterns in the temporal evaluation of sets of time-varying data. We follow Buchin~\etal~\cite{grouping2015} who define groups using three parameters: group-size, group-duration, and
inter-entity distance.  We give upper and lower bounds on the number of maximal groups over all
parameter values, and show how to compute them efficiently.  Furthermore, we
describe data structures that can report changes in the set of
maximal groups in an output-sensitive manner.  Our results hold in $\R^d$ for fixed $d$.
 \end{abstract}


\section{Introduction}
\label{sec:Introduction}

Time-varying phenomena are ubiquitous and hence the rapid increase in available tracking, recording, and storing technologies has led to an explosive growth in time-varying data. Such data comes in various forms: time-series (tracking a one-dimensional variable such as stock prices), two- or higher-dimensional trajectories (tracking moving objects such as animals, cars, or sport players), or ensembles (sets of model runs under varying initial conditions for one-dimensional variables such as temperature or rain fall), to name a few. Efficient tools to extract information from time-varying data are needed in a variety of applications, such as predicting traffic flow~\cite{lltx-dftf-10},
understanding animal movement~\cite{BovetB88}, coaching sports teams~\cite{fujimura2005dominating}, or forecasting the weather~\cite{Stohl1998947}. Consequently, recent years have seen a flurry of algorithmic methods to analyse time-varying data which can, for example, identify important geographical locations from a set of trajectories~\cite{BenkertDGW10,hotspots2013}, determine good average representations~\cite{bbklsww-mt-12}, or find patterns, such as groups traveling together~\cite{grouping2015,gudmundsson2007efficient,geogrouping2015}.

Most, if not all, of these algorithms use several parameters to model the applied problem at hand. The assumption is that the domain scientists, who are the users of the algorithm, know from years of experience which parameter values to use in their analysis. However, in many cases this assumption is not valid. Domain scientists do \emph{not} always know the correct parameter settings and in fact need algorithmic support to interactively explore their data in, for example, a visual analytics system~\cite{andrienko2007visualanalytics,keim2008}.

We present algorithms and data structures that support the interactive analysis of the grouping structure of one-, two-, or higher-dimensional time-varying data while varying all defining parameters. Grouping structures (which track the formation and dissolution of groups) characterise important patterns in the temporal evaluation of sets of time-varying data. Classic examples are herds of animals or groups of people. But also for one-dimensional ensembles grouping is meaningful, for example, when detecting trends in weather models~\cite{mirgazar2014boxplotensembles}.

Buchin~\etal~\cite{grouping2015} proposed a grouping structure for sets of moving entities. Their definition was later extended by Kostitsyna~\etal~\cite{geogrouping2015} to geodesic distances. In this paper we use the same trajectory grouping structure. Our contributions are data structures and query algorithms that allow the parameters of the grouping structure to vary interactively and hence make it suitable for explorative analysis of sets of time-varying data. Below we first briefly review the definitions of Buchin~\etal~\cite{grouping2015} and then state our contributions in detail.

\mypar{Trajectory grouping structure~\cite{grouping2015}} Let \X be a set of
$n$ entities moving in $\R^d$ and let \Time denote time. The entities trace
trajectories in $\Time \times \R^d$. We assume that each individual trajectory
is piecewise linear and consists of at most $\tau$ vertices.  Two entities $a$
and $b$ are \emph{$\eps$-connected} if there is a chain of entities
$a=c_1\codots c_k=b$ such that for any pair of consecutive entities $c_i$ and
$c_{i+1}$ the distance is at most $\eps$. A set $G$ is $\eps$-connected, if for
any pair $a,b \in G$, the entities are $\eps$-connected (possibly using
entities not in $G$). Given parameters $m$, $\eps$, and $\delta$, a set of
entities $G$ is an \emph{$(m,\eps,\delta)$-group} during time interval $I$ if
(and only if) (i) $G$ has size at least $m$, (ii) $\duration(I) \geq \delta$,
and (iii) $G$ is \emph{$\eps$-connected} at any time $t \in I$.  An
$(m,\eps,\delta)$-group $(G,I)$ is \emph{maximal} if $G$ is maximal in size or
$I$ is maximal in duration, that is, if there is no group $H \supset G$ that is
also $\eps$-connected during $I$, and no interval $J \supset I$ such that $G$
is $\eps$-connected during $J$.

\mypar{Results and Organization} We want to create a data structure \D that represents the grouping structure, that is, its maximal groups, while allowing us to efficiently change the parameters. 
As we show below, the complexity of the problem is already fully apparent for one-dimensional time-varying data. Hence we restrict our description to $\R^1$ in Sections 2--4 and then explain in Section~\ref{sec:Higher_Dimensions} how to extend our results to higher dimensions.

If all three parameters $m$, $\eps$, and $\delta$ can vary independently the question arises what constitutes a meaningful maximal group. Consider a maximal $(m,\eps,\delta)$-group $(G,I)$.  If we slightly increase
$\eps$ to $\eps'$, and consider a slightly longer time interval $I'
\supseteq I$ then $(G,I')$ is a maximal $(m,\eps',\delta)$-group.  Intuitively,
 these groups $(G,I)$ and $(G,I')$ are the same. Thus, we
are interested only in (maximal) groups that are ``combinatorially
different''. Note that the set of entities $G$ may also be a maximal
$(m,\eps,\delta)$-group during a time interval $J$ completely different from
$I$, we also wish to consider $(G,I)$ and $(G,J)$ to be combinatorially
different groups. In Section~\ref{sec:Groups} we formally define when two
(maximal) $(m,\eps,\delta)$-groups are (combinatorially) different. We prove that there are at most $O(|\A|n^2)$ such groups, where \A is the arrangement of the trajectories in
$\Time \times \R^1$, and $|\A|$ is its complexity. We also argue that the number of maximal groups may be as large as $\Omega(\tau n^3)$, even for fixed parameters $m$, $\eps$, and
$\delta$ and in $\R^1$. This significantly strengthens the lower bound of Buchin~\etal~\cite{grouping2015}.

In Section~\ref{sec:algo} we
present an $O(|\A|n^2 \log^2 n)$ time algorithm to compute all combinatorially
different maximal groups. In Section~\ref{sec:data_structures} we describe a
data structure that allows us to efficiently obtain all groups for a given set
of parameter values. Furthermore we also describe data structures for the
interactive exploration of the data. Specifically, given the set of maximal
$(m,\eps,\delta)$-groups we want to change one or more of the parameters and
efficiently report only those maximal groups which either ceased to be a
maximal group or became a maximal group. That is, our data structures can
answer so-called \emph{symmetric-difference queries} which are gaining in
importance as part of interactive analysis
systems~\cite{eppstein2013setdifference}. As mentioned above, in
Section~\ref{sec:Higher_Dimensions} we extend our data structures and
algorithms to $\R^d$, for fixed~$d$. 

\section{Combinatorially Different Maximal Groups}
\label{sec:Groups}

We consider entities moving in $\R^1$, hence the trajectories form an
arrangement \A in $\Time \times \R^1$. We assume that
no three pairs of entities have equal distance at the same time.
Consider the four-dimensional \emph{parameter space} \PP with axes time, size,
distance, and duration. A set of entities $G$ defines a region $A_G$ in this
space in which it is \emph{alive}: a point $p=(p_t,p_m,p_\eps,p_\delta)=(t,m,\eps,\delta)$ lies in $A_G$ if
and only if $G$ is a $(m,\eps,\delta)$-group at time $t$. We use these
regions to define when groups are combinatorially different. First (Section~\ref{sub:eps_Groups}) we fix $m=1$ and $\delta = 0$ and define and count the
number of combinatorially different maximal $(1,\eps,0)$-groups, over all
choices of parameter $\eps$. We then extend our results to include other values
of $\delta$ and $m$ in Section~\ref{sub:The_Number_of_Distinct_Maximal_Groups,_over_all_Parameters}.

\subsection{The Number of Distinct Maximal $(1,\eps,0)$-Groups, over
  all $\eps$}
\label{sub:eps_Groups}


Consider the $(t,\eps)$-plane in \PP through $\delta = 0$ and $m=1$. The
intersection of all regions $A_G$ with this plane give us the points $(t,\eps)$
for which $G$ is a $(1,\eps,0)$-group. Note that $G$ is a
$(1,\eps,0)$-group at time $t$ if and only if the set $G$ is $\eps$-connected
at time $t$. Hence the region $A_G$, restricted to this plane, corresponds to
the set of points $(t,\eps)$ for which $G$ is $\eps$-connected. $A_G$, restricted to this plane, is simply connected. Furthermore, as
the distance between any pair of entities moving in $\R^1$ varies linearly,
$A_G$ is bounded from below by a $t$-monotone polyline $f_G$. The region is
unbounded from above: if $G$ is $\eps$-connected (at time $t$) for some value
$\eps$, then it is also $\eps'$-connected for any $\eps' \geq \eps$ (see
Fig.~\ref{fig:example_polygons}).
Every maximal length segment in the intersection between (the restricted) $A_G$ and the
horizontal line $\ell_\eps$ at height $\eps$
corresponds to a (maximal) time
interval $I$ during which $(G,I)$ is a $(1,\eps,0)$-group, or an
\emph{$\eps$-group} for short. Every such a segment corresponds to an
\emph{instance} of $\eps$-group $G$.

\begin{figure}[tb]
  \centering
  \includegraphics{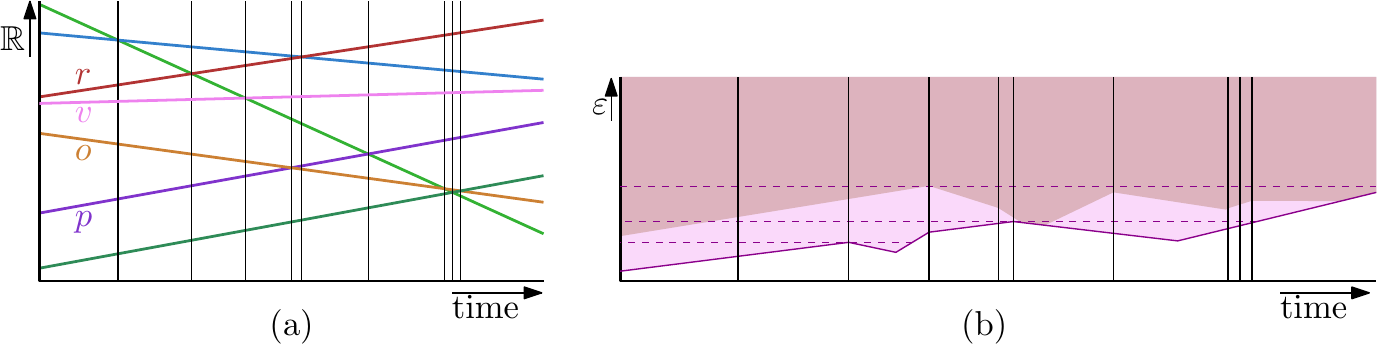}
  \caption{(a) A set of trajectories for a set of entities moving in $\R^1$ (b)
    The region $A_{\{r,v\}}$ during which $\{r,v\}$ is alive, and its
    decomposition into polygons, each corresponding to a distinct instance. In all
    such regions, except the top one $\{r,v\}$ is a maximal group: in the top
    region $\{r,v\}$ is dominated by $\{r,v,o\}$ (darker region).}
  \label{fig:example_polygons}
\end{figure}

\begin{observation}
  \label{obs:maximal_combinatorial_group}
  Set $G$ is a maximal $\eps$-group on $I$, iff the line segment
  $s_{\eps,I}=\{(t,\eps) \mid t \in I\}$ is a maximal length segment in $A_G$,
  and is not contained in $A_H$, for a supergroup $H \supset G$.
\end{observation}
Two instances of $\eps$-group $G$ may \emph{merge}.
Let $v$ be a local maximum of $f_G$ and $I_1 = [t_1,v_t]$ and $I_2 = [v_t,t_2]$ be two instances of group $G$ meeting at $v$.
At $v_\eps$, the two instances $G$ that are alive during $[t_1,v_t]$ and $[v_t,t_2]$ merge and we now
have a single time interval $I = [t_1,t_2]$ on which $G$ is a group.
We say that $I$ is a new instance of $G$, different from $I_1$ and $I_2$.
We can thus decompose $A_G$ into maximally-connected regions, each corresponding to a distinct instance of group $G$, using horizontal segments through the local maxima of $f_G$.
We further split each region at the values $\eps$ where $G$ changes between being maximal and being dominated.
Let $\P_G$ denote the obtained set of regions in which $G$ is maximal.
Each such a region $P$ corresponds to a \emph{combinatorially distinct} instance on which $G$ is a maximal group (with at least one member and duration at least zero).
The region $P$ is bounded by at most two horizontal line segments and two $\eps$-monotone chains (see Fig.~\ref{fig:example_polygons}(b)).


\mypar{Counting maximal $\eps$-groups}
To bound the number of distinct maximal $\eps$-groups, over all values of $\eps$, we have to count the number of polygons in $\P_G$ over all sets $G$.
While there are possibly exponentially many sets, there is structure in the regions $A_G$ which we can exploit.

Consider a set of entities $G$ and a region $P \in \P_G$ corresponding to a distinct instance of the maximal $\eps$-group $G$.
We observe that all vertices of $P$ lie on the polyline $f_G$: they are either directly vertices of $f_G$, or they are points $(t,\eps)$ on the edges of $f_G$ where $G$ starts or stops being maximal.
For the latter case there must be a polyline $f_H$, for some subgroup or supergroup of $G$, that intersects $f_G$ at such a point.
Furthermore, observe that any vertex (of either type) is used by at most a constant number of regions from $\P_G$.

Below we show that the complexity of the arrangement \H, of all polylines $f_G$ over all $G$, is bounded by $O(|\A|n)$.
Furthermore, we show that each vertex of \H can be incident to at most $O(n)$ regions.
It follows that the complexity of all polygons $P \in \P_G$, over all groups (sets) $G$, and thus also the number of such sets, is at most $O(|\A|n^2)$.

\begin{figure}[t]
  \centering
  \includegraphics{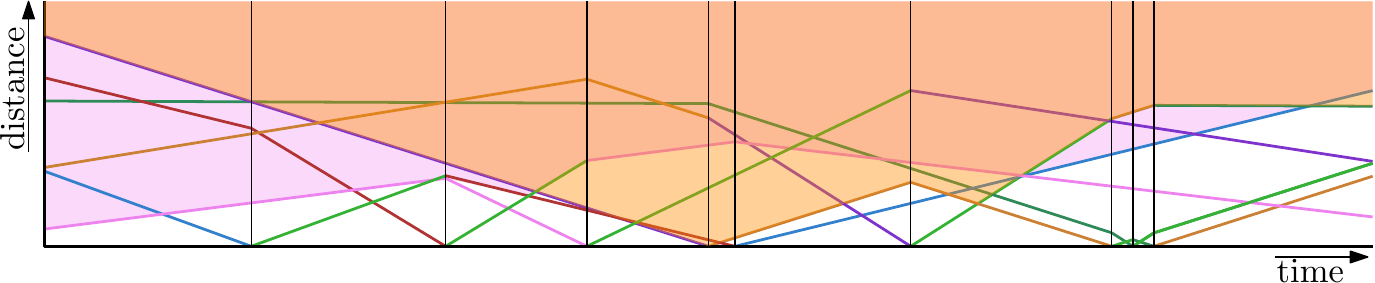}
  \caption{The arrangement \H and the regions $A_{\{r,v\}}$ (purple) and
    $A_{\{p,o\}}$ (orange) for the trajectories shown in
    Fig.~\ref{fig:example_polygons}(a). The arrangement \H corresponds to the
    arrangement of functions $h_a(t)$ that represent the distance from $a$ to
    the entity directly above $a$ at time $t$.}
  \label{fig:arrangement_H}
  \vspace{-\baselineskip}
\end{figure}

\mypar{The complexity of \H} The \emph{span} $S_G(t) = \{ a \mid a \in \X\
\land\ a(t) \in [\min_{b \in G} b(t), \max_{b \in G} b(t))\}$ of a set of
entities $G$ at time $t$ is the set of entities between the lowest and highest
entity of $G$ at time $t$ (for technical reasons, we include the lowest entity
of $G$ in the span, but not the highest). Let $h_a(t)$ denote the distance from
entity $a$ to the entity directly above $a$ at time $t$, that is, $h_a(t)$ is
the height of the face in \A that has $a$ on its lower boundary at time $t$.
\begin{observation}
  \label{obs:connected}
  A set $G$ is $\eps$-connected at time $t$, if and only if the largest
  distance among consecutive entities in $S_G(t)$ is at most $\eps$. That is,
  \vspace{-0.5\baselineskip}
  \[
    f_G(t) = \max_{a \in S_G(t)} h_a(t)
  \]
\end{observation}
It follows that \H is a subset of the
arrangement of the $n$ functions $h_a$, for $a \in \X$ (see
Fig.~\ref{fig:arrangement_H}). We use this fact to show that \H
has complexity at most $O(|\A|n)$:

\begin{lemma}
  \label{lem:upperbound_triplets}
  Let \A be an arrangement of $n$ line segments, and let $k$ be the maximum
  number of line segments intersected by a vertical line. The number of
  triplets $(F,F',x)$ such that the faces $F \in \A$ and $F' \in \A$ have equal
  height $h$ at $x$-coordinate $x$ is at most $O(|\A|k) \subseteq O(|\A|n)
  \subseteq O(n^3)$.
\end{lemma}

\begin{proof}
  Let $\ell_x$ be the vertical line through point $(x,0)$. Now consider a
  triplet $(F,F',x)$, and let $e_F$ and $f_F$ ($e_{F'}$ and $f_{F'}$) be the
  two edges of $F$ ($F'$) intersected by $\ell_x$. We charge $(F,F',x)$ to edge
  $e \in \{e_F, f_F, e_{F'}, f_{F'}\}$ if (and only if) its left endpoint, say
  $u$, is the rightmost endpoint that lies to the left of $\ell_x$ (i.e.~$u$ is
  the rightmost among the left endpoints). We now show that each edge can be
  charged at most $2k$ times.

  Consider an edge $e=\overline{uv}$ of \A, with $u_x \leq v_x$. Edge $e$ is
  charged by a triplet $(F,F',x)$, only if one of the faces, say $F$, is
  incident to $e$, and the left-endpoints of the three other edges that are
  intersected by $\ell_x$ and bounding $F$ or $F'$ lie to the left of $u$. It
  now follows that there are only $k$ choices for face $F'$, as both the edges
  bounding $F'$ are intersected (consecutively) by the vertical line through
  $u_x$, and any vertical line intersects at most $k$ edges. Clearly, $e$ is
  incident to at most two faces, and thus there are also only two choices for
  $F$. Finally, observe that for each such pair of faces $F$ and $F'$ there is
  at most one value $x \in [u_x,r]$, where $r$ is the $x$-coordinate of the
  leftmost right endpoint among $e_F$, $f_F$, $e_{F'}$ and $f_{F'}$, at which
  $F$ and $F'$ have equal height (as the height of faces $F$ and $F'$ varies
  linearly in such an interval). It follows that every edge $e \in \A$ is
  charged at most $2k$ times.
\end{proof}

\begin{remark}
  Interestingly, this bound is tight in the worst case. In
  Appendix~\ref{app:Lowerbound_Face_Heights} we give a construction where there
  are $\Omega(n^3)$ triplets $(F,F',x)$ such that $F$ and $F'$ have equal
  height at $x$, even if we use lines instead of line segments.
\end{remark}

\begin{lemma}
  \label{lem:complexity_H}
  The arrangement \H has complexity $O(|\A|n)$.
\end{lemma}

\begin{proof}
  Vertices in \H are either (i) vertices of individual functions $h_a$, or (ii)
  intersections between two such functions, say $h_a$ and $h_c$. The total
  complexity of the individual functions is $O(|\A|)$, hence there are also
  only $O(|\A|)$ vertices of the type (i). Vertices of the type (ii) correspond
  to a triplet $(F,F',x)$ in which $F$ and $F'$ are faces of \A that have equal
  height at $x$-coordinate $x$. By Lemma~\ref{lem:upperbound_triplets} there
  are at most $O(|\A|n)$ such triplets. Thus, the number of vertices of type
  (ii) is also at most $O(|\A|n)$.
\end{proof}

What remains to show is that each vertex $v$ of \H can be incident to at most
$O(n)$ polygons from different sets. We use Lemma~\ref{lem:per_vertex}, which follows from
Buchin \etal~\cite{grouping2015}:

\begin{lemma}
  \label{lem:per_vertex}
  Let \RG be the Reeb graph for a fixed value $\eps$ capturing the movement of
  a set of $n$ entities moving along piecewise-linear trajectories in $\R^d$
  (for some constant $d$), and let $v$ be a vertex of \RG. There are at most
  $O(n)$ maximal groups that start or end at $v$.
\end{lemma}


\begin{lemma}
  \label{lem:v_incident_to_groups}
  Let $v$ be a vertex of \H. Vertex $v$ is incident to at most $O(n)$ polygons
  from $\P = \bigcup_{G \subseteq \X} \P_G$.
\end{lemma}

\begin{proof}
  Let $P \in \P_G$ be a region that uses $v$. Thus, $G$ either starts or
  ends as a maximal $v_\eps$-group at time $v_t$. This means, $v$ correspond to
  a single vertex $u$ in the Reeb graph, built with parameter $v_\eps$. By
  Lemma~\ref{lem:per_vertex}, there are at most $O(n)$ maximal $v_\eps$-groups
  that start or end at $u$. Hence, $v$ can occur in regions of at most $O(n)$
  different sets $G$. For a fixed set $G$, the regions in $\P_G$
  are disjoint, so there are only $O(1)$ regions from $\P_G$, that contain
  $v$. 
\end{proof}
%
%
\begin{lemma}
  \label{lem:number_and_complexity_eps-grs}
  The number of distinct $\eps$-groups, over all values $\eps$, and the total
  complexity of all regions $\P = \bigcup_{G \subseteq \X} \P_G$, are both at most
  $O(|\H|n) = O(|\A|n^2)$.
\end{lemma}

\subsection{The Number of Distinct Maximal Groups, over all Parameters}
\label{sub:The_Number_of_Distinct_Maximal_Groups,_over_all_Parameters}

Maximal groups are
monotonic in $m$ and $\delta$ (see Buchin~\etal~\cite{grouping2015}); hence a maximal
$(m,\eps,\delta)$-group is also a maximal $(m',\eps,\delta')$-group for any
parameters $m' \leq m$ and $\delta' \leq \delta$. It follows that the number of
combinatorially different maximal groups is still at most $O(|\A|n^2)$.

\begin{figure}[tb]
  \centering
  \includegraphics{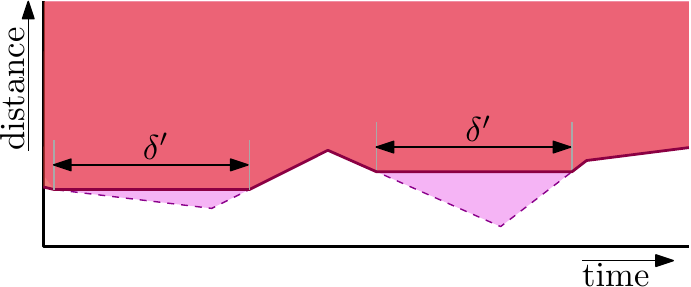}
  \caption{A cross section of the region $A_{\{r,v\}}$ with the plane through
    $\delta=\delta'$. The boundary of the original region (i.e. the cross
    section with the plane through $\delta=0$) is dashed.}
  \label{fig:blunted_group}
\end{figure}

For the complexity of the regions in $\bigcup \P_G$: fix $m = 0$, and consider
the remaining subspace of \PP with axes time, distance, and duration, and the
restriction of $A_G$, for any set $G$, into this space. In the $\delta = 0$
plane we simply have the regions $A_G$, that are bounded from below by a
$t$-monotone polyline $f_G$, as described in Section~\ref{sub:eps_Groups}. As
we increase $\delta$ we observe that the local minima in the boundary $f_G$ get
replaced by a horizontal line segment of width $\delta$ (see
Fig.~\ref{fig:blunted_group}). For arbitrarily small values
of $\delta > 0$, the total complexity of this boundary is still
$O(|\A|n^2)$. Further increasing $\delta$, monotonically decreases the number
of vertices on the functions $f_G$. It follows that the regions $A_G$,
restricted to the time, distance, duration space also have total complexity
$O(|\A|n^2)$. Finally, consider the regions $A_G$ in the full four dimensional
space. Clearly, $A_G \cap \{p \mid p \in \PP \land p_m < |G|\} = \emptyset$. For
values $m \geq |G|$, the boundary of $A_G$ is constant in $m$. We conclude:

\begin{theorem}
  \label{thm:num_max_distinct_groups}
  Let \X be a set of $n$ entities, in which each entity travels along a
  piecewise-linear trajectory of $\tau$ edges in $\R^1$, and let \A be the
  resulting trajectory arrangement. The number of distinct maximal groups is at
  most $O(|\A|n^2) = O(\tau n^4)$, and the total complexity of all regions in
  the parameter space corresponding to these groups is also $O(|\A|n^2) =
  O(\tau n^4)$.
\end{theorem}
In Section~\ref{sec:lower_bound} in the appendix we prove
Lemma~\ref{lem:lowerbound_fixed_eps}: even for fixed parameters $\eps$, $m$,
and $\delta$, the number of maximal $(m,\eps,\delta)$-groups, for entities
moving in $\R^1$, may be as large as
$\Omega(\tau n^3)$. This strengthens the result of
Buchin~\etal~\cite{grouping2015}, who established this bound for entities in $\R^2$.

\begin{lemma}
  \label{lem:lowerbound_fixed_eps}
  For a set \X of $n$ entities, in which each entity travels along a
  piecewise-linear trajectory of $\tau$ edges in $\R^1$, there can be
  $\Omega(\tau n^3)$ maximal $\eps$-groups.
\end{lemma}

\section{Algorithm}
\label{sec:algo}

In the following we refer to combinatorially different maximal groups simply as groups.
Our algorithm computes a representation (of size $O(|\A|n^2)$) of all groups, which we can use to list all groups and, given a pointer to a group $G$, list all its members and the polygon $Q_G \in \P_G$. We assume $\delta = 0$
and $m=1$, since the sets of maximal groups for $\delta > 0$ and $m > 1$ are a subset of the set for $\delta = 0$ and $m=1$.


\subsection{Overview}

Our algorithm uses the arrangement $\H$ located in the ($t,\eps$)-plane.
Line segments in \H correspond to the height function of the faces in \A.
Let $a,b \in S_G(t)$ be the pair of consecutive entities in the span of a group $G$ with maximum vertical distance at time $t$.
We refer to $(a,b)$ as the \emph{critical pair} of $G$ at time $t$.
The pair $(a,b)$ determines the minimal value of $\eps$ that is required for the group $G$ to be $\eps$-connected at time $t$.
The distance between a critical pair $(a,b)$ defines an edge of the polygon bounding $G$ in $\H$.

Our representation will consist of the arrangement \H in which each edge $e$ is
annotated with a data structure $\T_e$, a list \Lst (or array) with the top
edge in each \emph{group polygon} $Q_G \in \P_G$, and an additional data structure \S to support reconstructing the grouping polygons.
We start by computing the arrangement \H.
This takes $O(|\H|) = O(\tau n^3)$ time~\cite{amato2000computing}.
The arrangement is built from the set of height-functions of the faces of \A.
With each edge we store the pair of edges in \A responsible for it.

Given arrangement \H we use a sweep line algorithm to construct the rest of the representation.
A horizontal line $\ell(\eps)$ is swept at height $\eps$ upwards, and all
groups $G$ whose group polygon $Q_G$ currently intersects $\ell$ are maintained.
To achieve this we maintain a two-part status structure.
First, a set \S with for each group $G$ the time interval $I(G,\eps) = Q_G \cap \ell(\eps)$.
Second, for each edge $e \in \H$ intersected by $\ell(\eps)$ a data structure $\T_e$ with the sets of entities whose time interval starts or ends at $e$, that is, $G \in \T_e$ if and only if $I(G,\eps) = [s,t]$ with $s = e\cap \ell(\eps)$ or $t = e\cap \ell(\eps)$.
We postpone the implementation of $\T$ to Section \ref{sub:datastructure}.
The data structures support the following operations:
\\\\
\noindent
\begin{tabularx}{\textwidth}{p{0.22\textwidth}XX}
    \toprule
    Operation & Input        & Action \\
    \midrule
    \acall{Filter}{\T_e, X} & A data structure $\T_e$ \newline A set of entities $X$ & Create a data structure $\T' = \{ G \cap X \mid G \in \T_e\}$\\
    \acall{Insert}{\T_e, G} &  A data structure $\T_e$ \newline A pointer to a representation of $G$ & Create a data structure $\T' = \T_e \cup \{G\}$.\\
    \acall{Delete}{\T_e, G} &  A data structure $\T_e$ \newline A pointer to a representation of $G$ & Create a data structure $\T' = \T_e \setminus \{G\}$.\\
    \acall{Merge}{\T_e, \T_f} & Two data structures $\T_e$, $\T_f$, belonging to two edges $e,f$ having the same starting or ending vertex & Create a data structure $\T' = \T_e \cup \T_f$.\\
    \acall{Contains}{\T_e, G} &  A data structure $\T_e$ \newline A pointer to a representation of $G$ ending or starting on edge $e$ & Test if $\T_e$ contains set $G$.\\
    \acall{HasSuperSet}{\T_e, G} &  A data structure $\T_e$ \newline A pointer to a representation of $G$ ending or starting on edge $e$ & Test if $\T_e$ contains a set $H \supseteq G$, and return the smallest such set if so.\\
    \bottomrule
  \end{tabularx}
\def\arraystretch{1.0}


The end points of the time interval $I(G,\eps) = [\mathapply{start}{G,\eps}, \mathapply{end}{G,\eps}]$ vary non-stop along the sweep.
For each group $G$, the set \S instead stores the edges $e$ and $f$ of \H that contain the starting time \mathapply{start}{G,\eps} and ending time \mathapply{end}{G,\eps}, respectively, and pointers to the representation
of $G$ in $\T_e$ and $\T_f$.
We refer to $e$ and $f$ as the \emph{starting edge} and \emph{ending edge} of $G$.
In addition, we store with each interval $I(G,\eps)$ a pointer to the previous version of the interval $I(G,\eps')$ if (and only if) the starting time (ending time) of $G$ changed to edge $e$ (edge $f$) at $\eps'$.
Note that updates for both \S and \T occur only when a vertex is hit by the sweep line $\ell(\eps)$.
For all unbounded groups we add $I(G,\infty)$ to \Lst after the sweep line algorithm.

\begin{figure}[tb]
  \centering
  \includegraphics{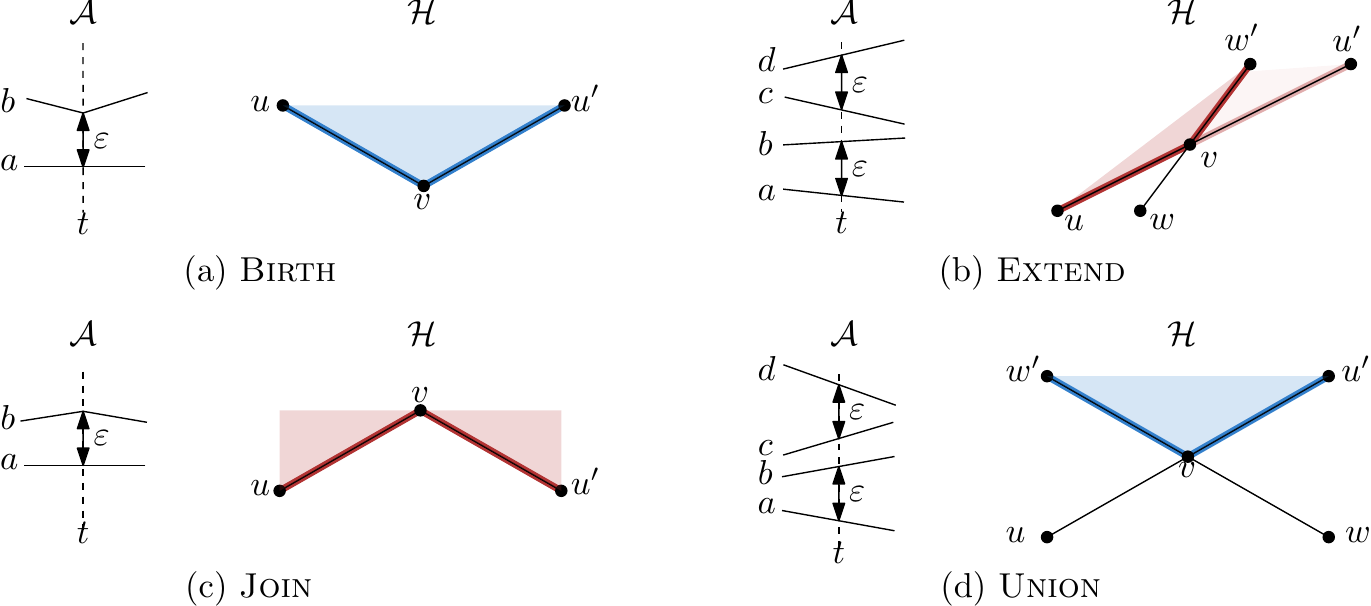}
  \caption{The different types of vertex events shown both in the arrangement \A and in \H. The {\sc Extend} event has a horizontally symmetric case.}
  \label{fig:vertex_roles}
\end{figure}


\subsection{Sweepline Events}
The sweep line algorithm results in four different vertex events (see
Fig.~\ref{fig:vertex_roles}). The \textsc{Extend}-event has a symmetrical
version in which $\overline{uu'}$ and $\overline{ww'}$ both have a negative
incline. We describe how to update our hypothetical data structures in all
cases.

%


\mypar{Case I - \textsc{Birth}} Vertex $v$ is a local minimum of one of the
functions $h_a$, with $a \in \X$ (see Fig.~\ref{fig:vertex_roles}(a)).
When the sweep line intersects $v$ a new maximal group $G$ is born.
We can find the maximal group spawned in $O(|G|)$ time by checking which trajectories are $\eps$-connected for this value of $t$ and $\eps$.
To this end we traverse the (vertical decomposition of) \A starting at the entities defining $v$.

\mypar{Case II - \textsc{Extend}} Vertex $v$ is the intersection of two
line segments $s_{ab} = \overline{uu'}$ and $s_{cd} = \overline{ww'}$, both
with a positive incline (see Fig.~\ref{fig:vertex_roles}(b)).
The case in which $s_{ab}$ and $s_{cd}$ have negative incline can be handled symmetrically.
Assume without loss of generality that $s_{cd}$ is steeper than $s_{ab}$.
We start with the following observation:

\begin{observation}
  \label{obs:no_bend}
  None of the groups arriving on edge $(w,v)$ continue on edge $(v,u')$.
\end{observation}

\begin{proof}
  Let $G$ be a group that arrives at $v$ using edge $(w,v)$.
  As $G$ uses $(w,v)$, it must contain entities both above and below the face $F$ defined by critical pair $(c,d)$.
  We know that $u_\eps$ and $w_\eps$ are strictly smaller than $v_\eps$ and $\eps$ is never smaller than zero.
  Thus, $v_\eps$ is strictly positive and $F$ has a strictly positive height at $t$.
  Therefore, $G$ still contains entities above and below $F$ after time $t$.
  But then the critical pair $(c,d)$ is still part of $G$ and $s_{cd}$ is a lower bound for the group.
  It follows that $G$ must use edge $(v,w')$.
\end{proof}

We first compute the groups on outgoing edge $(v,u')$.  By
Observation~\ref{obs:no_bend} all these groups arrive on edge $(u,v)$.  In
particular, they are the maximal size subsets from $\T_{(u,v)}$ for which all
entities lie below entity $d$ at time $t$, that is, $\T_{(v,u')} =
\acall{Filter}{\T_{(u,v)}, \below{d,t}}$, where $\below{y,t} = \{ x \mid x \in
\X \land x(t) < y(t)\}$.  For each group $G$ in $\T_{(v,u')}$ we update the
time-interval in \S.  If $G$ was dominated by a maximal group $H \supset G$ on
incoming edge $(u,v)$, we insert a new time interval with starting edge $f =
\gstart{H,\eps}$ and ending edge $(v,u')$ into \S, and insert $G$ into
$\T_f$. Note that $G$ and $H$ indeed have the same starting time: $G$ is a
subset of $H$, and is thus $\eps$-connected at any time where $H$ is
$\eps$-connected. Since $G$ was not maximal before, it did not start
earlier than $H$ either.

The groups from $\T_{(u,v)}$ that contain entities on both sides of critical pair $(c,d)$, continue onto edge $(v,w')$.
Let $\T'$
denote these groups.
We update the interval $I(G)$ in \S for each group $G \in \T'$ by setting the ending edge to $(v,w')$.

Next, we determine which groups from $\T_{(w,v)}$ die at $v$.
A maximal group $G \in \T_{(w,v)}$ dies at $v$ if there is a group $H$ on $(v,w')$ that dominates $G$.
Any such group $H$ must arrive at $v$ by edge $(u,v)$.
Hence, for each group $G \in \T_{(w,v)}$ we check if there is a group $H \in \T'$
with $H \supset G$ and $I(H) \supseteq I(G)$.
For each of these groups we remove the interval $I(G,\eps)$ from \S, add $I(G,\eps)$ to \Lst, and delete the set $G$ from the data structure $\T_f$, where $f$ is the starting edge of $G$ (at height $\eps$).

The remaining (not dominated) groups from $\T_{(w,v)}$ continue onto edge $(v,w')$.
Let $\T''$ denote this set.
We obtain $\T_{(v,w')}$ by merging $\T'$ and $\T''$, that is, $\T_{(v,w')} = \acall{Merge}{\T',\T''}$.
Since we now have the data structures $\T_{(v,u')}$ and $\T_{(v,w')}$, and we updated \S accordingly, our status structure again reflects the maximal groups currently intersected by the sweep line.





\mypar{Case III - \textsc{Join}} Vertex $v$ is a local maximum of one of
the functions $h_a$, with $a \in \X$ (see Fig.~\ref{fig:vertex_roles}(c)).
Two combinatorially different maximal groups $G_u$ and $G_w$ with the same set
of entities die at $v$ (and get replaced by a new maximal  group $G^*$) if and
only if $G_u$ is a maximal group in $\T_{(u,v)}$ and $G_w$ is a maximal group
in $\T_{(w,v)}$. We test this with a call to \acall{Contains}{\T_{(w,v)},G_u}
for each group $G_u \in \T_{(u,v)}$. Let $G$ be a group in $\T_{(u,v)}$, and
let $H \in \T_{(w,v)}$ be the smallest supergroup of $G$, if such a group exists.
At $v$ the group $G$ will immediately extend to the ending edge of $H$.
We can find $H$ by using a \acall{HasSuperSet}{\T_{(w,v)}, G} call.
If $H$ exists we insert $G$ into $\T_e$, and update $I(G,\eps)$ in \S accordingly.
We process the groups $G$ in $\T_{(w,v)}$ that have a group $H \in \T_{(u,v)}$ whose starting time jumps at $v$ analogously.


\mypar{Case IV - \textsc{Union}} Vertex $v$ is the intersection of a line segment $s_{ab} = \overline{uu'}$ with positive incline and a line segment $s_{cd} = \overline{ww'}$, with negative incline (see Fig.~\ref{fig:vertex_roles}(d)).
The \textsc{Union} event is a special case of the \textsc{Birth} event.
Incoming groups on edge $(u,v)$ are below the line segment $s_{cd}$ and, hence, can not contain any elements that are above $c$.
As a consequence the line segment $s_{cd}$ does not limit these groups and for a group $G \in \T_{(u,v)}$ we can safely add it to $\T_{(v,u')}$.
We also update the interval $I(G)$ in $\S$ by setting the ending edge to $(v,u')$.
An analogous argument can be made for groups arriving on edge $(w,v)$.

Furthermore a new maximal group is formed.
Let $H$ be the set of all entities $\eps$-connected to entity $a$ at time $t$.
We insert $H$ into $\T_{(v,u')}$ and $\T_{(v,w')}$ and we insert a time interval $I(H)$ into $\S$ with starting edge $(v,w')$ and ending edge $(v,u')$.

\subsection{Data Structure}\label{sub:datastructure}
We can implement \S
using any standard balanced binary search tree, the only requirement is that,
given a (representation of) set $G$ in a data structure $\T_e$, we can
efficiently find its corresponding interval in \S.

\subparagraph{The data structure $\T_e$.} We need a data structure $\T
= \T_e$ that supports \afunc{Filter}, \afunc{Insert}, \afunc{Delete},
\afunc{Merge}, \afunc{Contains}, and \afunc{HasSuperSet} efficiently.
%
We describe a structure of size $O(n)$, that supports \afunc{Contains} and
\afunc{HasSuperSet} in $O(\log n)$ time, \afunc{Filter} in $O(n)$ time, and
\afunc{Insert} and \afunc{Delete} in amortized $O(\log^2 n)$ time. In
general, answering \afunc{Contains} and \afunc{HasSuperSet} queries in a
dynamic setting is hard and may require $O(n^2)$ space~\cite{yellin1992subset}.

\begin{lemma}
  \label{lem:nested_sizes}
  Let $G$ and $H$ be two non-empty $\eps$-groups that both end at time $t$. We
  have:%
  \[
    (G \cap H \neq \emptyset \land |G| \leq |H|) \Longleftrightarrow G
    \subseteq H \land G \neq \emptyset.
  \]
\end{lemma}

\begin{proof}
  The if-direction is easy: $G \subseteq H$ immediately implies that $|G| \leq
  |H|$, and since $G$ is non-empty we then also have $G \cap H = G \neq
  \emptyset$.

  We prove the only-if direction by contradiction: assume by contradiction that
  $G \not\subseteq H$, and thus there is an element $b \in G \setminus
  H$. Furthermore, let $a \in G\cap H$, and let $s_G$ and $s_H$ denote the
  starting times of group $G$ and $H$, respectively. We distinguish two cases:
  $s_G \leq s_H$ and $s_H < s_G$.

  Case $s_G \leq s_H$. Since $a \in G$ and $s_G \leq s_H$, we have that at any
  time in $[s_H,z]$, the entities in $G$ are $\eps$-connected to $a$. So, in
  particular, entity $b$ is $\eps$-connected to $a$. However, during $[s_H,z]$,
  the entities in $H$ are also $\eps$-connected to $a$. Thus, it follows that
  during $[s_H,z]$, entity $b$ is also $\eps$-connected to $H$, and thus $b \in
  H$. Contradiction.

  Case $s_H < s_G$. Analogous to the previous case we get that both $H$ and $G$
  are $\eps$-connected to entity $a$ during $[s_G,z]$. It then follows that $H
  \subseteq G$. However, as $b \in G \setminus H$, this relation is strict,
  that is, $H \subset G$. This contradicts that $|G| \leq |H|$.
\end{proof}

%
%
We implement \T with a tree similar to the \emph{grouping-tree} used by Buchin~\etal~\cite{grouping2015}.
Let $\{G_1 \codots G_k\}$ denote the groups stored
in \T, and let $\X' = \bigcup_{i \in [1 \codots k]} G_i$ denote the entities in these
groups.
Our tree \T has a leaf for every entity in $\X'$.
Each group $G_i$ is represented by an internal node $v_i$.
For each internal node $v_i$ the set of leaves in the subtree rooted at $v_i$ corresponds exactly to the entities in $G_i$.
By Lemma~\ref{lem:nested_sizes} these sets indeed form a tree.
With each node $v_i$, we store the size of the group $G_i$, and (a pointer to) an arbitrary entity in $G_i$.
Next to the tree we store an array containing for each entity a pointer to the leaf in the tree that represents it (or \afunc{Nil} if the entity does not occur in any group).
We preprocess \T in $O(n)$ time to support level-ancestor (LA) queries as well as
lowest common ancestor (LCA) queries, using the methods of Bender and
Farach-Colton~\cite{bender2000lca, Bender20045}.
Both methods work only for \emph{static} trees, whereas we need to allow updates to \T as well.
However, as we need to query $\T_e$ only when processing the upper end vertex of $e$, we can be lazy in updating $\T_e$.
More specifically, we delay all updates, and simply rebuild $\T_e$ when we handle its upper end vertex.

\mypar{\afunc{HasSuperSet} and \afunc{Contains} queries}
Using LA queries we can do a binary search on the ancestors of a given node.
This allows us to implement both \acall{HasSuperSet}{\T_e,G} queries and \acall{Contains}{\T_e,G} in $O(\log n)$ time for a group $G$ ending or starting on edge $e$.
Let $a$ be an arbitrary element from group $G$.
If the datastructure $\T_e$ contains a node matching the elements in $G$ then it must be an ancestor of the leaf containing $a$ in $\T$.
That is, it is the ancestor that has exactly $|G|$ elements. By
Lemma~\ref{lem:nested_sizes} there is at most one such node.
As ancestors only get more elements as we move up the tree, we find this node in $O(\log n)$ time by binary search.
Similarly, we can implement the \afunc{HasSuperSet} function in $O(\log n)$ time.

\mypar{\afunc{Insert}, \afunc{Delete}, and \afunc{Merge} queries} The
\afunc{Insert}, \afunc{Delete}, and \afunc{Merge} operations on $\T_e$ are
performed lazily; We execute them only when we get to the upper vertex of edge
$e$. At such a time we may have to process a batch of $O(n)$ such
operations. We now show that we can handle such a batch in $O(n \log^2 n)$ time.


\begin{lemma}
  \label{lem:nested_large_groups}
  Let $G_1 \codots G_m$ be maximal $\eps$-groups, ordered by decreasing size, such
  that: (i) all groups end at time $t$, (ii) $G_1 \supseteq G_i$, for all $i$,
  (iii) the largest group $G_1$ has size $s$, and (iv) the smallest group has
  size $|G_m| > s/2$. We then have that $G_i \supseteq G_{i+1}$ for all $i \in
  [1 \codots m-1]$.
\end{lemma}

\begin{proof}
  All groups $G_i$ are subsets of $G_1$ and have size at least $s/2$. Thus, any
  two subsets $G_i$ and $G_{i+j}$ have a non-empty intersection, i.e. $G_i \cap
  G_{i+j} \neq \emptyset$. The result then follows directly from Lemma~\ref{lem:nested_sizes}.
\end{proof}

\begin{lemma}
  \label{lem:subset_test}
  Given two nodes $v_G \in \T$ and $v_H \in \T'$, representing the set $G$ respectively $H$, both ending at time $t$, we can test if $G \subseteq H$ in $O(1)$ time.
\end{lemma}

\begin{proof}
  Let $a$ be the entity from $G$ stored with $v_G$.
  We use the array of $\T'$ to find the leaf $\ell$ in $\T'$ that represents $a$, and perform a LCA query on $\ell$ and $v_H$ in $\T'$.
  If the result is $v_H$ then $a \in H$ and Lemma~\ref{lem:nested_sizes} states that $G \subseteq H$ if and only if $|G| < |H|$.
  If the result is not $v_H$ then $a \not\in H$, and trivially $G \not\subseteq H$.
  Finding $\ell$, and performing the LCA query takes $O(1)$ time.
  As we store the group size with each node, we can also test if $|G| < |H|$ in constant time.
\end{proof}

\begin{lemma}
  \label{lem:mass_merge_in_nlog2n}
  Given $m = O(n)$ nodes representing maximal $\eps$-groups $G_1 \codots G_m$,
  possibly in different data structures $\T_1 \codots \T_m$, that all share ending
  time $t$, we can construct a new data structure \T representing $G_1 \codots G_m$
  in $O(n \log^2 n)$ time.
\end{lemma}


\noindent
{\color{darkgray}\sffamily\bfseries Proof.}
  Sort the groups $G_1 \codots G_m$ on decreasing group size.
  Let $G_1 \in \T_1$ denote the largest group and let it have size $s$.
  We assume for now that $G_1$ is a superset  of all other groups.
  If this is not the case we add a dummy group $G_0$ containing all elements.
  We process the groups in order of decreasing size.
  By Lemma~\ref{lem:nested_large_groups} it follows that all groups $G_1 \codots G_k$ that are larger than $s/2$ form a path $P$ in \T, rooted at $G$.

\begin{wrapfigure}{r}{0pt}
    \raggedleft
    \includegraphics{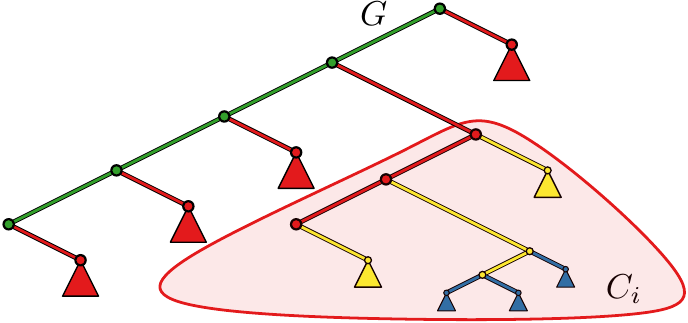}
    \caption{$\T$ is built top-down in several rounds.
    Edges and nodes are colored by round.}
    \label{fig:mass_merge}
\end{wrapfigure}
  For all remaining (small) groups $G_i$ we then find the smallest group in $P$ that is a super set of $G_i$.
  By Lemma~\ref{lem:subset_test}, we can test in $O(1)$ time if a group $H \in
  P$ is a supergroup of $G_i$ by performing a LCA query in the tree $H$
  originated from. We can then find the smallest super set of $G_i$ in $O(\log
  n)$ time using a binary search.
  Once all groups are partitioned into clusters with the same ancestor $G_i$, we process the clusters recursively.
  When the largest group in a cluster has size one we are done (see Fig.~\ref{fig:mass_merge}).

  The algorithm goes through a series of rounds.
  In each round the remaining clusters are handled recursively.
  Because all (unhandled) clusters jointly contain no more than $O(n)$ groups, each round takes only $O(n \log n)$ time in total.
  As in each round the size of the largest group left is reduced by half, it follows that after $O(\log n)$ rounds the algorithm must has constructed the complete tree.
  Updating the array with pointers to the leaves takes $O(n)$ time, as does rebuilding the tree for future LA and LCA queries.
\qed
\smallskip

\noindent
The final function \afunc{Filter} can easily be implemented in linear time by pruning the tree from the bottom up.
We thus conclude:

\begin{lemma}
  \label{lem:handle_event}
  We can handle each event in $O(n \log^2 n)$ time.
\end{lemma}

\subsection{Maximal Groups}

\mypar{Reconstructing the grouping polygons} Given a group $G$, represented
by a pointer to the top edge of $Q_G$ in \Lst, we can construct the complete
group polygon $Q_G$ in $O(|Q_G|)$ time, and list all group members of $G$ in
$O(|G|)$ time.
We have
access to the top edge of $Q_G$. This is an interval $I(G,\hat{\eps})$ in \S,
specifically, the version corresponding to $\hat{\eps}$, where $\hat{\eps}$ is
the value at which $G$ dies as a maximal group. We then follow the pointers to
the previous versions of $I(G,\cdot)$ to construct the left and right chains of
$Q_G$. When we encounter the value $\check{\eps}$ at which $G$ is born, these
chains either meet at the same vertex, or we add the final bottom edge of $Q_G$
connecting them. To report the group members of $G$, we follow the pointer to
$I(G,\hat{\eps})$ in \S. This interval stores a pointer to its starting edge
$e$, and to a subtree in $\T_e$ of which the leaves represent the entities in
$G$.

\mypar{Analysis} The list \Lst contains $O(g) = O(|\A|n^2)$ entries
(Theorem~\ref{thm:num_max_distinct_groups}), each of constant size. The total
size of all \S's is $O(|\H|n)$: at each vertex of \H, there are only a linear
number of changes in the intervals in \S. Each edge $e$ of \H
stores a data structure $\T_e$ of size $O(n)$. It follows that our
representation uses a total of $O(|\H|n) = O(|\A|n^2)$ space. Handling each of
the $O(|\H|)$ nodes requires $O(n\log^2 n)$ time, so the total running time is
$O(|\A|n^2\log^2 n)$.

\begin{theorem}
  \label{thm:compute_combinatorial_max_groups}
  Given a set \X of $n$ entities, in which each entity travels along a
  trajectory of $\tau$ edges, we can compute a representation of all $g =
  O(|\A|n^2)$ combinatorial maximal groups \G such that for each group in \G we
  can report its grouping polygon and its members in time linear in its
  complexity and size, respectively. The representation has size $O(|\A|n^2)$
  and takes $O(|\A|n^2 \log^2 n)$ time to compute, where $|\A| = O(\tau n^2)$ is the
  complexity of the trajectory arrangement.
\end{theorem}

\section{Data Structures for Maximal Group Queries}
\label{sec:data_structures}

In this section we present data structures that allow us to efficiently obtain
all groups for a given set of parameter values
(Section~\ref{sub:Quering_the_maximal_groups}), and for the interactive
exploration of the data (Section~\ref{sub:Symmetric_Difference}). Throughout
this section, $n$ denotes the number of entities considered, $\tau$ the number
of vertices in any trajectory, $k$ the output complexity, i.e.~the number of
groups reported, $g$ the number of maximal groups, $g'$ the maximum number of
maximal groups for a given (fixed) value of $\eps$, and $\Pi$ the total
complexity of the regions corresponding to the $g$ combinatorially different
maximal groups. So we have $g'=O(\tau n^3)$ and $g\leq \Pi = O(\tau n^4)$. When
$g'$, $g$, or $\Pi$ appear as the argument of a logarithm, we write $O(\log
\taun)$.

\subsection{Quering the maximal groups}
\label{sub:Quering_the_maximal_groups}

We show that we can store all
groups in a data structure of size $O(\Pi\log \taun \log n)$ that can be built in
$O(\Pi \log^2 \taun\log n )$ time, and allows reporting all $(m,\eps,\delta)$-groups
in $O( \log^2 \taun\log n + k)$ time.
We use the following three-level tree to achieve this.

\begin{wrapfigure}{r}{0pt}
  \raggedleft
  \includegraphics{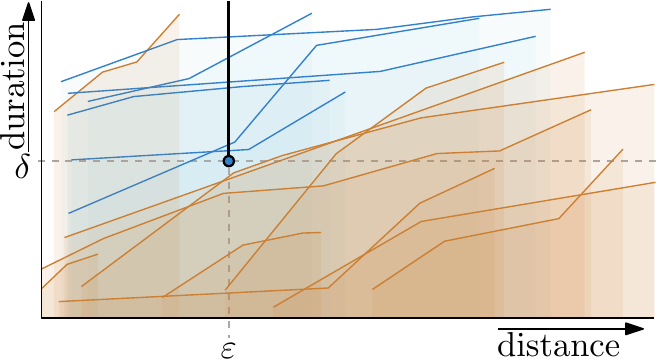}
  \caption{The functions $D_G$ expressing the duration of group $G$ as a
    function of $\eps$. Assuming all groups have size at least $m$, all
    $(m,\eps,\delta)$-groups intersect the upward vertical half-ray starting in
    point $(\eps,\delta)$.}
  \label{fig:x-monotone_polyline_stabbing}
\end{wrapfigure}
On the first level we have a balanced binary tree with in the leaves the group
sizes $1...n$. Each internal node $v$ corresponds to a range $R_v$ of
group sizes and stores all groups whose size lies in the range $R_v$. Let
$\G_v$ denote this set of groups, and for each such group let $D_G$ denote
the duration of group $G$ as a function of $\eps$. The functions
$D_G$ are piecewise-linear, $\delta$-monotone, and may intersect (see
Fig.~\ref{fig:x-monotone_polyline_stabbing}). By
Theorem~\ref{thm:num_max_distinct_groups} the total complexity of
these functions is $O(\Pi)$.
We store all functions $D_G$, with $G \in \G_v$, in
a data structure that can answer the following \emph{polyline stabbing}
queries in $O(\log^2 \taun + k)$ time:
Given a query point $q=(\eps,\delta)$, report all polylines that pass above point $q$,
that is, for
which $D_G(\eps) \geq \delta$. Thus, given parameters $m$, $\eps$, and
$\delta$, finding all $(m,\eps,\delta)$-groups takes $O(\log^2 \taun\log n + k)$ time.

We build a segment tree storing the ($\eps$-extent of the) individual edges
of all polylines stored at $v$. An internal node $u$ of the segment tree
corresponds to an interval $I(u)$, and stores the set of edges
$\mathapply{Ints}{u}$ that completely span $I(u)$. Hence, with respect to $u$,
we can consider these segments as lines. For a query with a point $q$, we
have to be able to report all (possibly intersecting) lines from
$\mathapply{Ints}{u}$ that pass above $q$. We use a duality transform to
map each line $\ell$ to a point $\ell^*$ and query point $q$ to a line
$q^*$. The problem is then to report all points $\ell^*$ in the half-plane
below $q^*$. Such queries can be answered in $O(\log h + k)$ time, using $O(h)$
space and $O(h \log h)$ preprocessing time, where $h$ is the number of points
stored~\cite{chazelle1985duality}. It follows that we can find all $k$
polylines that pass above $q$ in $O(\log^2 \taun + k)$ time, using $O(\Pi \log \taun)$
space, and $O(\Pi \log^2 \taun)$ preprocessing time. We thus obtain the following
result:

\begin{theorem}
  \label{thm:query_groups}
  Given parameters $m$, $\eps$, and $\delta$, we can build a data structure of
  size $O(\Pi\log \taun \log n)$, using $O(\Pi  \log^2 \taun\log n)$ preprocessing time,
  which can report all $(m,\eps,\delta)$-groups in $O( \log^2 \taun\log n + k)$
  time, where $k$ is the output complexity.
\end{theorem}

\subsection{Symmetric Difference Queries}
\label{sub:Symmetric_Difference}

Here we describe data structures for the interactive exploration of the
data. We often have all $(m,\eps,\delta)$-groups, for some parameters $m$,
$\eps$, and $\delta$, and we want to change (some of the) parameters, say to
$m'$, $\eps'$, and $\delta'$, respectively. Thus, we need a way to efficiently
report all changes in the maximal groups. This requires us to solve
\emph{symmetric difference queries}, in which we want to efficiently report all
maximal $(m,\eps,\delta)$-groups that are no longer maximal for parameters
$m'$, $\eps'$, and $\delta'$, and all maximal $(m',\eps',\delta')$-groups that
were not maximal for parameters $m$, $\eps$, and $\delta$. That is, we wish to
report $\G(m,\eps,\delta) \symd \G(m',
\eps',\delta')$.

\mypar{Changing only $\delta$} Consider the case in which we vary only
$\delta$, and keep $m$ and $\eps$ fixed, that is, $m'=m$ and
$\eps'=\eps$. With fixed $\eps$, it suffices to use the
algorithm from Buchin~\etal~\cite{grouping2015} to compute all maximal
$\eps$-groups with size at least $m$. There are at most $g'$ such groups.
Each group $G$ corresponds to an interval $I_G =
(-\infty,\mathapply{duration}{G}]$ such that $G$ is
maximal for a choice $\delta$ of the duration parameter if and only if
$\delta\in (-\infty,\mathapply{duration}{G}]$.
We now have two values $\delta$ and $\delta'$, and we should report all
intervals in $S(\delta) \symd S(\delta')$, where $S(x)$ denotes the intervals
that contain $x$.

Note that we can assume
without loss of generality that $\delta \leq \delta'$. Then we
observe that we should report group $G$ if and only if
$\mathapply{duration}{G}\in [\delta,\delta']$. Hence, our data structure
is simply a balanced binary search tree on at most $g'$ values
$\mathapply{duration}{G}$ and a symmetric difference query is
a 1-dimensional range query.

\mypar{Changing only $m$} The case in which we vary only $m$ can be solved
analogously to the previous case. A maximal group $G$ has a size $|G|$, and
$G$ should be reported if and only if $|G|\in [m,m')$, assuming that the group size
changes from $m$ to $m'$ or vice versa, with $m<m'$.


\begin{wrapfigure}{r}{0pt}
    \centering
    \includegraphics{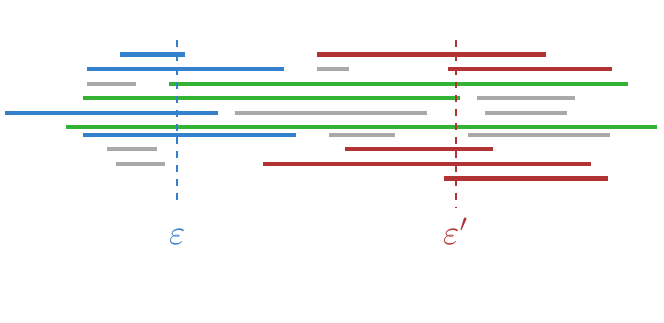}
	\caption{The symmetric difference for parameters $\eps$ and $\eps'$ (red and
    blue intervals) is exactly the set of intervals that contains either $\eps$
    or $\eps'$, but not both. Hence, the green intervals should not be
    reported.}\label{fig:symm_diff_intervals}
\end{wrapfigure}
\mypar{Changing only $\eps$} The minimum duration $\delta$ is fixed, so
consider the $\delta$-truncated grouping polygons (i.e.~the regions $A_G$ where
each local minimum has been replaced by a horizontal line segment of width
$\delta$). Compute all combinatorially distinct maximal groups for this
parameter $\delta$ and remove all groups that have size less than~$m$.  A
group $G$ is now maximal during some interval $I_G =
[\check{\eps}_G,\hat{\eps}_G]$, and we have to report $G$ if (and only if)
$I_G$ occurs in the set $S(\eps) \symd S(\eps')$. We now observe that this is
the case exactly when $I_G$ contains $\eps$ or $\eps'$, but not both (see
Fig.~\ref{fig:symm_diff_intervals}). Using an interval tree we can thus report
the symmetric difference for $\eps$ in $O(\log \taun + k)$ time, using $O(\Pi)$
space and $O(\Pi \log \taun)$ preprocessing time.\footnote{Note that we now
  have $O(\Pi)$ groups (intervals) rather than $O(g')$.}

\mypar{Changing $\delta$ and $m$ simultaneously} Consider the space $\delta
\times m$. A group $G$ is maximal in the quadrant
$(-\infty,\mathapply{duration}{G}] \times (-\infty,|G|]$ with top-right corner
$p_G = (\mathapply{duration}{G},|G|)$.  So, for parameters $\delta$ and $m$,
the set of maximal $(m,\eps,\delta)$-groups corresponds to the set of
corner points that lie to the top-right of $(\delta,m)$.
It now follows that when we change the parameters to
$(\delta',m')$, the maximal groups that we have to report lie in
$Q_{(\delta,m)} \symd Q_{(\delta',m')}$ (see
Fig.~\ref{fig:symmetric_difference_delta_m}). We can report those points
(groups) by two three-sided (orthogonal) range queries. Therefore, we store the
corner points in a priority search tree~\cite{bkos-cgaa-08}, and thus solve
symmetric difference queries in $O(\log \taun + k)$ time, and $O(g')$
space. Building a priority search tree takes $O(g' \log \taun)$ time.

\begin{figure}[tb]
  \centering
  \includegraphics{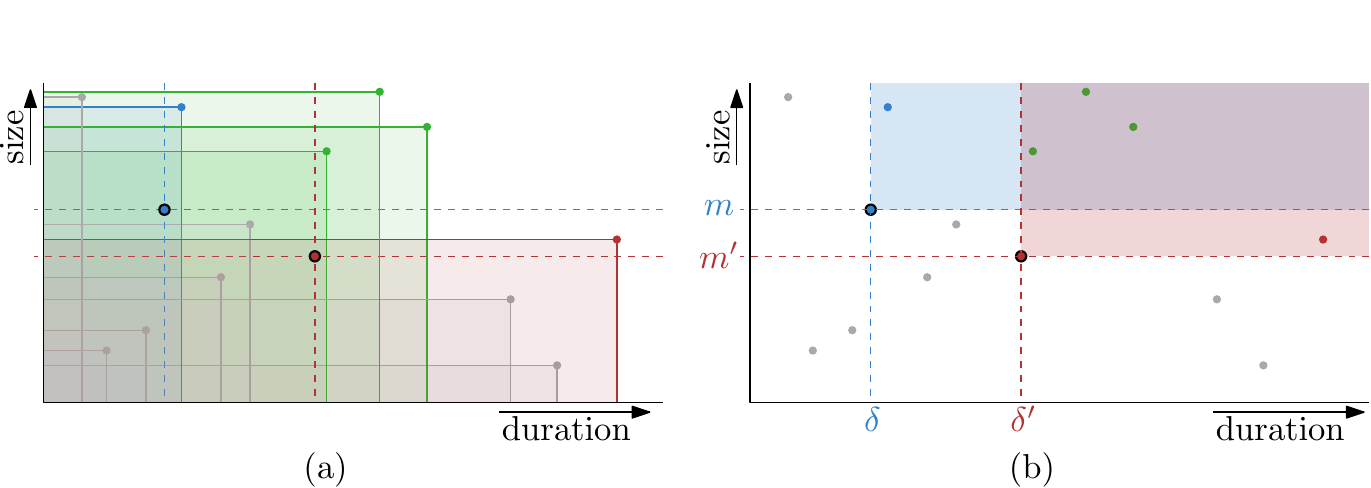}
  \caption{Symmetric difference queries when we allow varying $\delta$ and
    $m$. Each combinatorial maximal group $G$ corresponds to a lower-left
    quadrant in the space $\delta \times m$ (a). For given
    parameters $m$, $\eps$, and $m$, all $(\delta,\eps,m)$-groups lie to the
    top-right of the point $(\delta,m)$. Therefore, the groups that have to be
    reported in a symmetric-difference query (shown in red and blue) can be
    reported by two three-sided range queries.}
  \label{fig:symmetric_difference_delta_m}
\end{figure}

\mypar{Changing $\eps$ and $m$ simultaneously} Consider the space $\eps
\times m$. A group $G$ is now a maximal group in a bottomless rectangle
$R_G=[\check{\eps}_G,\hat{\eps}_G] \times (-\infty,|G|]$. See
Fig.~\ref{fig:symd_eps_m}. Thus, for parameters $\eps$ and $m$ the maximal
groups all contain the point $(\eps,m)$. We find the groups
that we have to report by combining the approaches for varying only $\eps$ and
varying only $m$.
Observe that $G$ should be reported if (and only if) $(\eps, m)$ is in the
rectangle $R_G$ and $(\eps',m')$ is not, or vice versa. Assume we test for the former.
We can solve this query problem
with three very similar two-level data structures.
The first is a binary search tree on all groups $G$ sorted on $\check{\eps}_G$.
An internal node $v$ is associated to a subset $\G_v$ of groups that appear in
the subtree rooted at $v$. We store $\G_v$ by storing the horizontal line
segments $[\check{\eps}_G,\hat{\eps}_G] \times |G|$ in a hive graph~\cite{Chazelle86},
preprocessed for planar point location queries.
If $h=|\G_v|$, then this structure uses $O(h)$ storage and and allows us
to report all line segments of $\G_v$ that lie vertically above a query point
in $O(\log h + k)$ time. We query the main tree with $\eps'$ and select a
subset of $O(\log \taun)$ nodes whose associated subsets contains exactly the
groups $G$ with $\eps'< \check{\eps}_G$. This implies that $(\eps',m')$ is
not inside $R_G$. The second-level structure allows us to find those groups
whose rectangle $R_G$ contains $(\eps,m)$.
The second data structure is different only in its main tree, which is sorted on
$\hat{\eps}_G$, and we will select the nodes whose associated subsets
contains exactly the groups with $\eps'>\hat{\eps}_G$. The third data structure
is again different in the main tree only, and is sorted on $|G|$. Here we select
nodes whose associated subsets have $m'>|G|$. Together, the three main trees
capture that $(\eps',m')$ is not in $R_G$ and the associated structures
capture that $(\eps,m)$ is in $R_G$. We report any group in the symmetric difference
at most twice.
The data structure uses $O(\Pi \log \taun)$ space and takes $O(\Pi \log^2 \taun)$ time
to build. The query time is $O(\log^2 \taun + k)$.

\begin{wrapfigure}[16]{r}{0.54\textwidth}
  \vspace{-.5\baselineskip}
  \centering
  \includegraphics{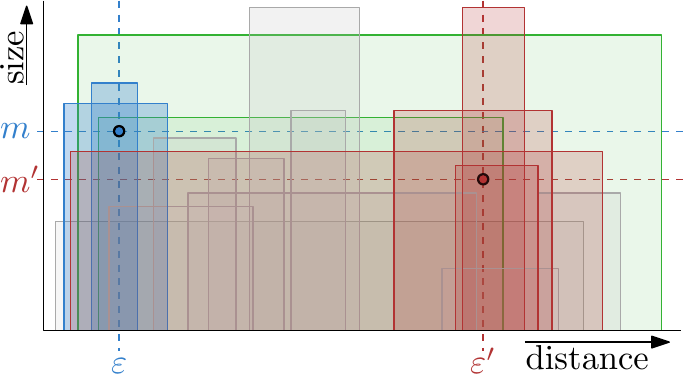}
  \caption{Symmetric difference queries when we allow varying $\eps$ and $m$.
    Each combinatorial maximal group $G$ now corresponds to a bottomless
    rectangle. We find the groups that are maximal for only one pair of
    parameters (i.e. the red and blue groups) by a query in a two-level tree
    for symmetric-difference queries. }
  \label{fig:symd_eps_m}
\end{wrapfigure}

\mypar{Changing $\eps$ and $\delta$, one by one} Consider the space $\eps
\times \delta$. A group $G$ is maximal in the region below the partial,
piecewise-linear, and monotonically increasing function $D_G$ that expresses
the duration of $G$ as a function of $\eps$. Each such partial function is
defined for a single interval of $\eps$-values. See
Fig.~\ref{fig:x-monotone_polyline_stabbing}.
Note that the polylines representing $D_G$ and $D_H$, for two groups $G$ and
$H$, may intersect. The combination of non-orthogonal boundaries and
intersections makes changing $\eps$ and $\delta$ much harder than changing
$\eps$ and $m$.

Consider changing parameter $\delta$ to $\delta'$, while keeping $\eps$
unchanged. For such a query we thus have to report all groups $G$ whose
polyline $D_G$ intersects the vertical query segment
$Q=\overline{(\eps,\delta)(\eps,\delta')}$. We use the following data structure
to answer such queries. We build a segment tree storing the individual edges of
the polylines. Like in Section~\ref{sec:data_structures},
each node $v$ in this tree now corresponds to a set $\mathapply{Ints}{v}$
of polyline edges (one per polyline) that completely cross the interval $I_v$
associated with $v$. We again treat these edges as lines. We store the
$h=|\mathapply{Ints}{v}|$ lines in a data structure by Cheng and
Janardan~\cite{cheng1992intersectionsearching} that has size $O(h \log h)$, can
be built in $O(h\log^2 h)$ time, and allows reporting all (possibly
intersecting) lines that intersect $Q$ in $O(\sqrt{h}2^{\log^*h}\log h + k)$
time. Since for any value $\eps$ there are at most $O(g')$ groups, we also have
that for any node $v$, $|\mathapply{Ints}{v}| = O(g')$. It follows that we can
answer symmetric difference queries in $\delta$ in $O(\sqrt{g'}2^{\log^*
  \taun}\log^2\taun + k)$ time, after $O(\Pi \log^3 \taun)$ preprocessing time,
and using $O(\Pi \log^2 \taun)$ space.

Consider changing parameter $\eps$ to $\eps'$, while keeping $\delta$ unchanged.
For such a query we have to report all groups $G$ whose
polyline $D_G$ is above exactly one end point of the horizontal query segment
$Q=\overline{(\eps,\delta)(\eps',\delta)}$. Since all polylines $D_G$ are
$\eps$ and $\delta$-monotone we could use the same approach as before,
reversing the roles of $\eps$ and $\delta$.
However, a horizontal line may intersect $O(g)$ polylines
rather than $O(g')$, causing $\sqrt{g}$ to appear in the query time rather
than $\sqrt{g'}$. This may be significantly worse.
Instead, observe that there are three ways in which $Q$ can have exactly one
end point below $D_G$. The two cases where one end point of $Q$ is outside the
$\eps$-range of $D_G$ are easily handled with a two-level tree. The first level
is a binary search tree on the $\eps$-range of $D_G$, and allows us to find the
groups for which $D_G$ is either defined completely before, or completely after
$\eps$. All these groups are \emph{not} maximal for parameter $\eps$, so among
them we have to select the ones that are maximal for parameters $\eps'$ and
$\delta'$. Our second level is thus the data structure from
Section~\ref{sec:data_structures}. This leads to a data structure of size $O(\Pi \log^2 \taun)$
and query time $O(\log^3 \taun + k)$. The third case concerns the situation where the
$\eps$-range of $Q$ is contained in the $\eps$-range of $D_G$. In that case we need
to test whether $Q$ intersects $D_G$.
We use a hereditary segment tree~\cite{hereditary} on the $\eps$-ranges of all
segments $S$ of all $D_G$,
and at each node $v$, we use associated structures for the cases where segments of
$S$ are ``long'' and $Q$ is ``short'', and vice versa.
For the segments $S_v$ of $S$ that are long at $v$,
we observe that there are only $O(g')$ of them, because they have a common
$\eps$-value. Furthermore, there can be at most one long segment in $S_v$ for each
group $G$. Hence, we can use the data structure by Cheng and
Janardan~\cite{cheng1992intersectionsearching} to report the ones intersecting $Q$.
For the segments of $S$ that are short at $v$, we know that the query segment is long
and horizontal, so we can just consider the $\delta$-span of each short segment.
However, we must still ensure that the polyline $D_G$ that a short segment is part
of, extends to the right beyond $Q$. Both conditions together lead again to a hive
graph, preprocessed for planar point location.
The data structure has size $O(\Pi\log^2 \taun)$, query time
$O(\sqrt{g'}2^{\log^* \taun}\log^2\taun +k)$, and can be built in $O(\Pi\log^3 \taun)$ time.

\mypar{Changing all three parameters, one by one} To support changing all
three parameters, we combine some of the previous approaches. We build two
separate data structures; one to change $m$, the other to change $\eps$ or
$\delta$. The data structure to change $m$ is simply the data structure from
Section~\ref{sec:data_structures}. The first level of this tree
allows us to find $O(\log n)$ subtrees containing the groups whose size is in
the range $(\min\{m,m'\},\max\{m,m'\}]$. We then use the associated data
structures to report the groups that are long enough (with respect to
parameters $\eps$ and $\delta$). Thus, we can answer such queries in $O(\log n
\log^2 \taun + k)$ time. To support changing $\eps$ or $\delta$ we extend the
solution from the previous case: we simply add an other level to the structure,
that allows us to filter the groups that intersect a query segment $Q$ in the
$(\eps,\delta)$-plane by size. This yields a query time of
$O(\sqrt{g'}2^{\log^* \taun}\log^2\taun \log n + k)$. The size and preprocessing
time remain unaffected, when compared to the previous situation.

\mypar{Changing $\eps$ and $\delta$ simultaneously} We build a data
structure that allows us to report the maximal groups for parameters $\eps$ and
$m$ as a small number of canonical subsets. For each of these canonical subsets
we store a data structure that allows us to efficiently report the groups that
are \emph{not} maximal for parameters $\eps'$ and $\delta'$. Symmetrically, we
find the groups that are not maximal for parameters $\eps$ and $\delta$ and maximal for
$\eps'$ and $\delta'$. So, basically, we need two layers of the data structure
from Section~\ref{sec:data_structures}.\footnote{We described this
  data structure for reporting all maximal groups for $\eps$ and $\delta$. But it
  is easy to see that we can also use it to report all groups that are not
  maximal for $\eps$ and $\delta$: we simply have to report all polylines below,
  rather than above, point $(\eps,\delta)$.}

Recall that the data structure from
Section~\ref{sec:data_structures} is a segment tree with associated
data structures that allow half-plane range reporting. Unfortunately, the data
structure that we use for the half-space range reporting does not report the
result as a small number of canonical subsets. So, for the first layer of our
data structure we replace these half-plane range reporting data structures by a
partition tree~\cite{matousek1992partitiontrees}. For the second layer we can
use the data structure from Section~\ref{sec:data_structures} as
is. It follows that we can now find all groups that are maximal for $\eps$ and
$\delta$ but not maximal for $\eps'$ and $\delta'$ in $O(\sqrt{g'}2^{\log^* \taun}\log^3
\taun + k)$ time. The data structure uses $O(\Pi \log^2 \taun)$ space, and can be built
in $O(\Pi \log^3 \taun)$ time. We can thus solve symmetric difference queries in the
same time (and with the same amount of space).

\mypar{Changing all three parameters simultaneously} We use the same
approach as above, expressing the groups alive for parameters $m$, $\eps$, and
$\delta$ as a small number of canonical subsets, for each of which we store the
data structure from Section~\ref{sec:data_structures}. It follows
that we can report symmetric difference queries in $O(\sqrt{g'}2^{\log^*
  \taun}\log^3 \taun\log^2 n + k)$ time, using $O(\Pi \log^2 \taun\log^2 n)$ space and
  $O(\Pi \log^3 \taun\log^2 n)$ preprocessing time.

The following theorem summarizes our results:

\begin{theorem}
  \label{thm:symmetric_difference_queries}
  Let $(m,\eps,\delta)$ and $(m',\eps',\delta')$ be two configurations of
  parameters. In $O(P(\Pi,g',n))$ time we can build a data structure of size
  $O(S(\Pi,g',n))$ for symmetric difference queries, that is, we can report all
  groups in $\G(\eps,m,\delta) \symd \G(\eps',m',\delta')$, in $O(Q(\Pi,g',n,k))$
  time. In these results $\Pi$ denotes the total complexity of all
  combinatorially different maximal groups (over all values $\eps$), $g'$ the
  number of maximal groups for a fixed value $\eps$, $n$ the number of
  entities, $\tau$ the number of vertices in a trajectory, and $k$ the output
  complexity. The functions $P$, $S$, and $Q$
  depend on which of the parameters are allowed to change (other
  parameters are assumed to be fixed and known at preprocessing time). We have
  \\\\
  \begin{tabularx}{\textwidth}{ccXX}
    \toprule
    Variable Param. & Query time $Q(\Pi,g',n,k)$        & Space $S(\Pi,g',n)$ & Preproc. $P(\Pi,g',n)$ \\
    \midrule
    \multicolumn{4}{l}{Changing one parameter at a time} \\
    \midrule
    $\delta$           & $\log \taun + k$                    & $g' \log \taun$  & $g'\log \taun$\\
    $m$                & $\log \taun + k$                    & $g' \log \taun$  & $g'\log \taun$\\
    $\eps$             & $\log \taun  + k$                    & $\Pi \log \taun$  & $\Pi\log \taun$\\
    $\delta, m$        & $\log \taun + k$                    & $g' \log \taun$  & $g'\log \taun$\\
    $\eps, m$          & $\log^2 \taun + k$                  & $\Pi \log^2  \taun$ & $\Pi \log^2 \taun$\\
    $\eps, \delta$     & $\sqrt{g'}2^{\log^* \taun}\log^2 \taun + k$   & $\Pi \log^2\taun$ & $\Pi \log^3 \taun$ \\
    $\eps, \delta, m$  & $\sqrt{g'}2^{\log^* \taun}\log^2\taun \log n +
    k$ & $\Pi \log^2\taun $    & $\Pi \log^3 \taun$ \\
    \midrule
    \multicolumn{4}{l}{Changing multiple parameters at the same time} \\
    \midrule
    $\delta, m$       & $\log \taun + k$                              & $g' \log  \taun$ & $g'\log \taun$\\
    $\eps, m$         & $\log^2 \taun + k$                            & $\Pi \log^2  \taun$ & $\Pi\log^2 \taun$\\
    $\eps, \delta$    & $\sqrt{g'}2^{\log^* \taun}\log^3 \taun + k$         & $\Pi \log^2 \taun$ & $\Pi \log^3 \taun$ \\
    $\eps, \delta, m$ & $\sqrt{g'}2^{\log^* \taun}\log^3 \taun\log^2 n + k$ & $\Pi \log^2 \taun\log^2 n$ & $\Pi \log^3 \taun\log^2 n$ \\
    \bottomrule
  \end{tabularx}
\end{theorem}


\section{Entities Moving in $\R^d$}
\label{sec:Higher_Dimensions}

We now describe how our results can be
extended to entities moving in $\R^d$, for any constant dimension $d$.

\subsection{Bounding the Complexity}
\label{sub:Complexity}

Recall that $\X(t)$ denotes the locations of the entities at time $t$.
We still consider the $(t,\eps)$-plane in \PP, and the regions $A_G$, for
subsets $G \subseteq \X$, in which set $G$ is alive. Such a region is still
bounded from below by a function $f_G$ that expresses the minimum epsilon
value for which $G$ is $\eps$-connected. We again consider the arrangement \H
of these functions $f_G$, over all sets $G$.

Let $\H''$ be the arrangement of all pairwise distance functions $h_{ab}(t) =
\|a(t)b(t)\|$. For any subset of entities $G$ and any
time $t$, $f_G(t) = \|a(t)b(t)\|$ for some pair of entities $a$ and $b$. Thus,
\H is a sub-arrangement of $\H''$. This immediately gives an $O(\tau n^4)$ bound
on the complexity of \H. We instead show that \H
has complexity at most $O(\tau n^3\beta_4(n))$, where $\beta_s(n) =
\lambda_s(n)/n$, and $\lambda_s(n)$ is the maximum length of a
Davenport-Schinzel sequence of order $s$ on $n$ symbols.
Using exactly the same argument as in
Lemma~\ref{lem:v_incident_to_groups} we then get a bound of $O(\tau n^4\beta_4(n))$ on the number of combinatorially
different groups.

Let $\E(t)$ be the Euclidean minimum spanning tree (EMST) of the points in
$\X(t)$. We then observe:

\begin{observation}
  \label{obs:emst}
  A subset of entities $G \subseteq \X$ is $\eps$-connected at time $t$ if and
  only if for any two entities $p,q \in G$ the longest edge in the Euclidean minimum spanning tree $\E(t)$ on the path between $p$ and $q$ has length at most $\eps$.
\end{observation}

\begin{wrapfigure}{r}{0.37\textwidth}
  \raggedleft
  \includegraphics{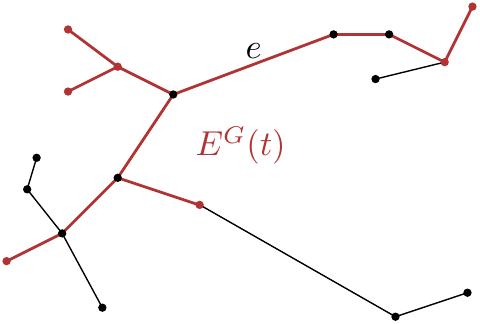}
  \caption{$\E(t)$ and its minimum subtree $E^G(t)$
    (red edges) for a subset of entities $G$ (red vertices).
    Edge $e$ determines the
    minimum $\eps$ for which $G$ is $\eps$-connected at~$t$.}
  \label{fig:emst}
\end{wrapfigure}
\noindent
Specifically, let $\E^G(t)$ be the minimum (connected) subtree of $\E(t)$
containing all points in $G(t)$, and let $\check{\eps}_G(t)$ be the length of
the longest edge $e$ in $\E^G(t)$ (see Fig.~\ref{fig:emst}). We have that $G$ is an
$\eps$-group for all $\eps \geq \check{\eps}_G(t)$, and that $f_G(t) =
\max_{(a,b) \in \E^G(t)} \|a(t)b(t)\|$.

It follows from Observation~\ref{obs:emst} that we are interested in the
distance function $h_{ab}$ on the time intervals during which $(a,b)$ is part
of the EMST. Hence, we need to consider only the arrangement of such partial
functions. It is, however, difficult to bound the complexity of the resulting
arrangement directly. Instead, we consider $h_{ab}$ only during those time
intervals in which $(a,b)$ is an edge in the Yao-graph~\cite{y-cmstk-82}. Let $\H'$ be
the resulting arrangement. Since the EMST is a subgraph of the Yao-graph it
follows that \H is a sub-arrangement of $\H'$~\cite{y-cmstk-82}.
%

\begin{lemma}
  \label{lem:complexity_Hp}
  $\H'$ has complexity $O(\tau n^3\beta_4(n))$.
\end{lemma}

\begin{proof}
  Fix an entity $a$, and consider the movement of the other entities with
  respect to $a$. This gives us a set of (piecewise linear) trajectories.
  Entity $a$ is fixed at the origin. Partition this
  space into $k = O(1)$ equal size polyhedral cones $C_1\codots C_k$ that have
  their common apex at the origin\footnote{Note that $k$ is
    exponential in the dimension $d$.}. For each such cone $C_i$, let
  $\eta^i_a(t)$ denote the distance from $a$ to the nearest entity in the
  cone. It is easy to show that $\eta^i_a$ is piecewise hyperbolic, and
  consists of $O(\tau \lambda_4(n))$ pieces~\cite{geogrouping2015}.

  Let $\H^*$ be the arrangement of all functions $\eta^i_a$, over all entities
  $a \in \X$ and all cones $C_i$. The total number of pieces
  (hyperbolic arcs), over all entities and all cones, is $O(\tau
  n\lambda_4(n))$. Partition time into $O(\tau \lambda_4(n))$ time intervals,
  with $O(n)$ pieces each. This may require splitting some of the pieces, but
  the total number of pieces remains $O(\tau n\lambda_4(n))$. In each time
  interval we now have $O(n)$ hyperbolic arc pieces, that intersect at most
  $O(n^2)$ times in total. It follows that $\H^*$ has complexity $O(\tau
  \lambda_4(n)n^2) = O(\tau n^3\beta_4(n))$.

  Fix a time $t$, and consider the graph $Y(t)$ that has an edge $(a,b)$ if and
  only if $b$ is the nearest neighbor of $a$ in one of the cones $C_i$ at time
  $t$, that is, $\|a(t)b(t)\| = \eta^i_a(t)$. Indeed, $Y(t)$ is the Yao-graph
  of the entities at time $t$~\cite{y-cmstk-82}. It follows that $\H^* = \H'$.
\end{proof}
Since \H is a sub arrangement of $\H'$, it follows that \H also has complexity
$O(\tau n^3\beta_4(n))$. Using exactly the same argument as in
Lemma~\ref{lem:v_incident_to_groups} we then get that the number of
combinatorially different maximal groups is $O(\tau n^4\beta_4(n))$. We
conclude:

\begin{theorem}
  \label{thm:max_group_higher_dimensions}
  Let \X be a set of $n$ entities, in which each entity travels along a
  piecewise-linear trajectory of $\tau$ edges in $\R^d$, for any constant
  $d$. The number of maximal combinatorial groups as well as the total
  complexity of all their group polygons is at most $O(\tau n^4\beta_4(n))$.
\end{theorem}

\subsection{Algorithm}
\label{sub:Algorithm}

We can almost directly apply our algorithm from Section~\ref{sec:algo} in higher dimensions as well.
Instead of the arrangement \H, we now use $\H'$. The only
differences involve discovering the set entities involved in a
\textsc{Birth}-event, and splitting the set of entities in case of an
\textsc{Extend}-event. Let $v$ denote the vertex of $\H'$ we are
processing.
We use that at time $v_t$, $\H'$
encodes the Yao-graph $Y$.
Using a breadth first search in $Y$ we can find the entities connected to $v$.
If an edge has length larger than $\eps$ we stop the exploration along it.
Since $Y$ is planar, and has $O(n)$ vertices this takes $O(n)$
time. This does not affect the running time, hence we get the same result as in
Theorem~\ref{thm:compute_combinatorial_max_groups} for entities moving in
$\R^d$, for any constant $d$.

\subsection{Data Structures}
\label{sub:Data_Structures}

\mypar{Finding all maximal $(m,\eps,\delta)$-groups} We use the same
approach as in Section~\ref{sec:data_structures}. However, the
functions $\mathfunc{duration}_G$ are no longer (piecewise) linear functions in
$\eps$.
Let $\mathfunc{start}_G(\eps) = f^{-1}(\eps)$ and
$\mathfunc{end}_G(\eps) = h^{-1}(\eps)$ be some hyperbolic functions $f$ and
$h$ corresponding to curves in \H.
We have that $\mathfunc{duration}_G(\eps) = \mathfunc{end}_G(\eps) -
\mathfunc{start}_G(\eps)$. The function $\mathfunc{duration}_G$
corresponds to a piecewise curve with pieces defined by polynomials of
degree at most four. Hence, we have to solve the following sub-problem: given a
set of $g'$ algebraic curves of degree at most four, and query point $q$,
report all curves that pass above $q$.

We can solve such queries as follows. We transform the curves into hyperplanes
in $\R^\ell$, where $\ell$ is the linearization dimension. We then apply a
duality transform, after which the problem can be solved using a half-space
range query. Since we have curves of degree at most four in $\R^2$, the
linearization dimension is seven: the set of points above a curve can be
described using a seven-variate polynomial (the five coefficients of the degree
four curve, and the two coordinates of the
point)~\cite{agarwal1994semialgebraic_rangesearching}. It follows that we can
find all curves above query point $q$ in $O(g'^{1-1/\lfloor 7/2\rfloor}\polylog
\taun + k) = O({g'}^{2/3}\polylog \taun + k)$ time, using linear
space~\cite{matousek1992partitiontrees}. Reporting all maximal
$(m,\eps,\delta)$-groups thus takes $O({g'}^{2/3}\polylog \taun + k)$ time,
using $O(\Pi \log \taun\log n)$ space and $O(\Pi \log^2 \taun\log n)$
preprocessing time.

Alternatively, we can maintain the upper envelope of the curves in a dynamic
data structure. To solve a query, we repeatedly delete the curve realizing the
upper envelope at $q_\eps$. This allows us to find all $(m,\eps,\delta)$-groups
in $O(k\beta_4(g')^2\polylog \taun)$ time~\cite{chan2010dynamicconvexhull}.

\mypar{Symmetric Difference Queries} Only the versions of the problem that
involve changing both $\eps$ and $\delta$ are affected. Instead of piecewise
linear functions $\mathfunc{duration}_G$ we again have piecewise curves of
degree at most four. We use a similar approach as above to find the curves that
intersect a vertical or horizontal query segment in $O({g'}^{2/3}\polylog \taun
+ k)$ time. Thus, we essentially replace the $\sqrt{g'}$ terms in
Theorem~\ref{thm:symmetric_difference_queries} by a $g'^{2/3}$ term. 



\small
\section*{Acknowledgments}

F.S. is supported by the Danish National Research Foundation under grant
nr.~DNRF84. F.S. is supported by the Danish National Research Foundation under
grant nr.~DNRF84. A.v.G. and B.S. are supported by the Netherlands Organisation
for Scientific Research (NWO) under project nr. 612.001.102 and 639.023.208,
respectively.

\bibliography{cake}

\begin{thebibliography}{10}

\bibitem{agarwal1994semialgebraic_rangesearching}
P.~Agarwal and J.~Matoušek.
\newblock {On Range Searching with Semialgebraic Sets}.
\newblock {\em Disc. \& Comput. Geom.}, 11(4):393--418, 1994.

\bibitem{amato2000computing}
N.~Amato, M.~Goodrich, and E.~Ramos.
\newblock Computing the arrangement of curve segments: Divide-and-conquer
  algorithms via sampling.
\newblock In {\em Proc. 11th ACM-SIAM Symp. on Disc. Algorithms}, pages
  705--706, 2000.

\bibitem{andrienko2007visualanalytics}
G.~Andrienko, N.~Andrienko, and S.~Wrobel.
\newblock Visual analytics tools for analysis of movement data.
\newblock {\em ACM SIGKDD Explorations Newsletter}, 9(2):38--46, 2007.

\bibitem{bender2000lca}
M.~Bender and M.~Farach-Colton.
\newblock The {LCA} problem revisited.
\newblock In {\em LATIN 2000: Theoret. Informatics}, volume 1776 of {\em LNCS},
  pages 88--94. Springer, 2000.

\bibitem{Bender20045}
M.~Bender and M.~Farach-Colton.
\newblock The level ancestor problem simplified.
\newblock {\em Theoret. Computer Science}, 321(1):5--12, 2004.

\bibitem{BenkertDGW10}
M.~Benkert, B.~Djordjevic, J.~Gudmundsson, and T.~Wolle.
\newblock Finding popular places.
\newblock {\em Int. J. of Comput. Geom. \& Appl.}, 20(1):19--42, 2010.

\bibitem{BovetB88}
P.~Bovet and S.~Benhamou.
\newblock Spatial analysis of animals' movements using a correlated random walk
  model.
\newblock {\em J. of Theoret. Biology}, 131(4):419--433, 1988.

\bibitem{bbklsww-mt-12}
K.~Buchin, M.~Buchin, M.~van Kreveld, M.~L{\"o}ffler, R.~Silveira, C.~Wenk, and
  L.~Wiratma.
\newblock Median trajectories.
\newblock {\em Algorithmica}, 66(3):595--614, 2013.

\bibitem{grouping2015}
K.~Buchin, M.~Buchin, M.~van Kreveld, B.~Speckmann, and F.~Staals.
\newblock Trajectory grouping structure.
\newblock {\em J. of Comput. Geom.}, 6(1):75--98, 2015.

\bibitem{chan2010dynamicconvexhull}
T.~Chan.
\newblock A dynamic data structure for 3-d convex hulls and 2-d nearest
  neighbor queries.
\newblock {\em J. of the ACM}, 57(3):16:1--16:15, 2010.

\bibitem{Chazelle86}
B.~Chazelle.
\newblock A functional approach to data structures and its use in
  multidimensional searching.
\newblock {\em {SIAM} J. Comput.}, 17(3):427--462, 1988.

\bibitem{hereditary}
B.~Chazelle, H.~Edelsbrunner, L.~Guibas, and M.~Sharir.
\newblock Algorithms for bichromatic line-segment problems and polyhedral
  terrains.
\newblock {\em Algorithmica}, 11(2):116--132, 1994.

\bibitem{chazelle1985duality}
B.~Chazelle, L.~Guibas, and D.~Lee.
\newblock The power of geometric duality.
\newblock {\em BIT Numerical Mathematics}, 25(1):76--90, 1985.

\bibitem{cheng1992intersectionsearching}
S.~Cheng and R.~Janardan.
\newblock Algorithms for ray-shooting and intersection searching.
\newblock {\em Journal of Algorithms}, 13(4):670--692, 1992.

\bibitem{bkos-cgaa-08}
M.~de~Berg, O.~Cheong, M.~van Kreveld, and M.~Overmars.
\newblock {\em Computational Geometry: Algorithms and Applications}.
\newblock Springer, Berlin, 3rd edition, 2008.

\bibitem{eppstein2013setdifference}
D.~Eppstein, M.~Goodrich, and J.~Simons.
\newblock Set-difference range queries.
\newblock In {\em Proc. 2013 Canadian Conf. on Comput. Geom.}, 2013.

\bibitem{fujimura2005dominating}
A.~Fujimura and K.~Sugihara.
\newblock Geometric analysis and quantitative evaluation of sport teamwork.
\newblock {\em Systems and Computers in Japan}, 36(6):49--58, 2005.

\bibitem{gudmundsson2007efficient}
J.~Gudmundsson, M.~van Kreveld, and B.~Speckmann.
\newblock Efficient detection of patterns in {2D} trajectories of moving
  points.
\newblock {\em GeoInformatica}, 11:195--215, 2007.

\bibitem{hotspots2013}
J.~Gudmundsson, M.~van Kreveld, and F.~Staals.
\newblock Algorithms for hotspot computation on trajectory data.
\newblock In {\em Proc. 21st ACM SIGSPATIAL GIS}, pages 134--143, 2013.

\bibitem{keim2008}
D.~Keim, G.~Andrienko, J.-D. Fekete, C.~Görg, J.~Kohlhammer, and G.~Melançon.
\newblock Visual analytics: Definition, process, and challenges.
\newblock In A.~Kerren, J.~Stasko, J.-D. Fekete, and C.~North, editors, {\em
  Information Visualization}, volume 4950 of {\em LNCS}, pages 154--175.
  Springer, 2008.

\bibitem{geogrouping2015}
I.~Kostitsyna, M.~van Kreveld, M.~L{\"o}ffler, B.~Speckmann, and F.~Staals.
\newblock Trajectory grouping structure under geodesic distance.
\newblock In {\em Proc. 31th Symp. Computat. Geom.} Lipics, 2015.

\bibitem{lltx-dftf-10}
X.~Li, X.~Li, D.~Tang, and X.~Xu.
\newblock Deriving features of traffic flow around an intersection from
  trajectories of vehicles.
\newblock In {\em Proc.\ IEEE 18th Int.\ Conf.\ Geoinformatics}, pages 1--5,
  2010.

\bibitem{matousek1992partitiontrees}
J.~Matoušek.
\newblock Efficient partition trees.
\newblock {\em Disc. \& Comput. Geom.}, 8(3):315--334, 1992.

\bibitem{mirgazar2014boxplotensembles}
M.~Mirzargar, R.~Whitaker, and R.~Kirby.
\newblock {Curve Boxplot:} generalization of {Boxplot} for ensembles of curves.
\newblock {\em IEEE Trans. on Vis.\ and Comp.\ Graphics}, 20(12):2654--2663,
  2014.

\bibitem{Stohl1998947}
A.~Stohl.
\newblock Computation, accuracy and applications of trajectories -- a review
  and bibliography.
\newblock {\em Atmospheric Environment}, 32(6):947--966, 1998.

\bibitem{y-cmstk-82}
A.~Yao.
\newblock On constructing minimum spanning trees in $k$-dimensional spaces and
  related problems.
\newblock {\em SIAM J. Comput.}, 11(4):721--736, 1982.

\bibitem{yellin1992subset}
D.~Yellin.
\newblock Representing sets with constant time equality testing.
\newblock {\em J. of Algorithms}, 13(3):353--373, 1992.

\end{thebibliography}

\clearpage

\appendix

\section{The Number of Equal Height Faces in an Arrangement of Lines}
\label{app:Lowerbound_Face_Heights}

Recall that \A was an arrangement of $n$ line segments, and that $S(\A)$
denotes the set of all triples $(F,F',x)$ such that (\textit{i}) the faces $F
\in \A$ and $F' \in \A$ have equal height $h$ at $x$-coordinate $x$, and
(\textit{ii}) all faces in between $F$ and $F'$ at $x$-coordinate $x$ have
height less than $h$. We now show that $S(\A)$ may contain $\Omega(n^3)$
triples, even if our segments are lines.

\begin{lemma}
  \label{lem:lowerbound_equal_cell_heights}
  The number of triples in $S(\A)$ for a line arrangement \A with $n$ lines may
  be $\Omega(n^3)$.
\end{lemma}

\begin{figure}[b]
    \centering
    \includegraphics[page=1]{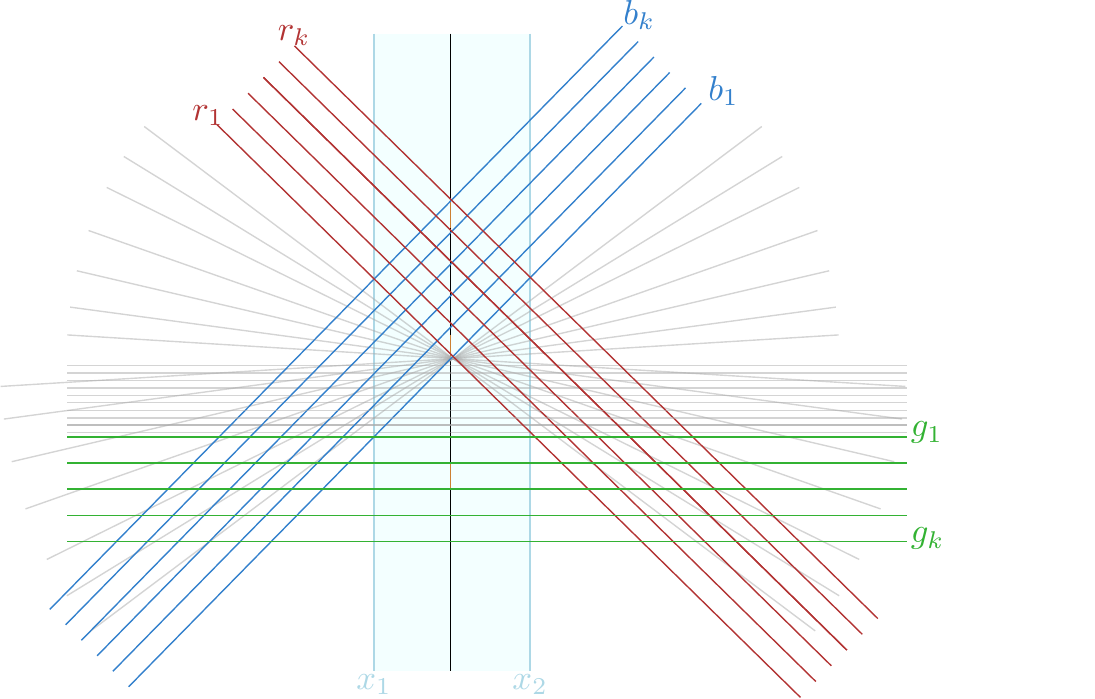}
    \caption{A construction of $n$ lines, such that their arrangement \A has
      $|S(\A)| = \Omega(n^3)$.}
    \label{fig:lowerbound_equal_cellheights}
\end{figure}

\begin{proof}
  We construct a set of $n$ lines $L = R \cup G \cup B \cup E $
   whose
  arrangement \A has $|S(\A)| = \Omega(n^3)$. The (sub)sets $R$, $G$, and $B$
  have size $k$ each. We use them to build the subset of faces such that there
  are $\Omega(k^3)$ triples $(F,F',x)$ such that $F$ and $F'$ have equal height
  at $x$. The remaining $O(k)$ lines are used only to make sure that the faces
  in between any such pair $(F,F')$ have smaller height. It follows that we can
  choose $k = \Theta(n)$ and get $|S(\A)| = \Omega(n^3)$ as desired.

  Our construction is shown in Fig.~\ref{fig:lowerbound_equal_cellheights}. The
  set of red lines $R$ together with the set of blue lines $B$ form a unit grid
  that has been rotated by $45+\delta$ degrees, for some arbitrarily small
  $\delta > 0$. The lines in $R$ ($B$) are all parallel to each
  other\footnote{Note that we can perturb all lines slightly to avoid parallel
    lines if desired.}. Let $C_{i,j}$ denote the face (grid cell) in which the
  intersection point of $r_i$ and $b_j$ is the point with the minimum
  $y$-coordinate, and let $\C_j =
  \{C_{1,j},C_{2,j+1},C_{3,j+2}..,C_{1+k_j,j+k_j}\}$ be the $j^\text{th}$
  ``column'' of such faces.

  The green lines in the set $G=\{g_1 \codots g_k\}$, ordered from top to bottom,
  are (almost) horizontally, and such that the distance between any consecutive
  lines $g_i$ and $g_{i+1}$ increases slightly (i.e.~the distance between $g_i$
  and $g_{i+1}$ is $h + i\delta'$, for some $h \in (1/c,1)$, some constant $c$,
  and some arbitrarily small $\delta'$). 
  We place these green lines sufficiently
  far below the grid formed by the red and blue lines such that each face
  $F_i$,
   bounded by $r_1, b_1$, and the green lines $g_i$ and $g_{i+1}$, is
  wide enough such that its $x$-extent contains $[x_1,x_2]$, the $x$-extent of
  the grid.

  \begin{figure}[t]
    \centering
    \includegraphics[page=2]{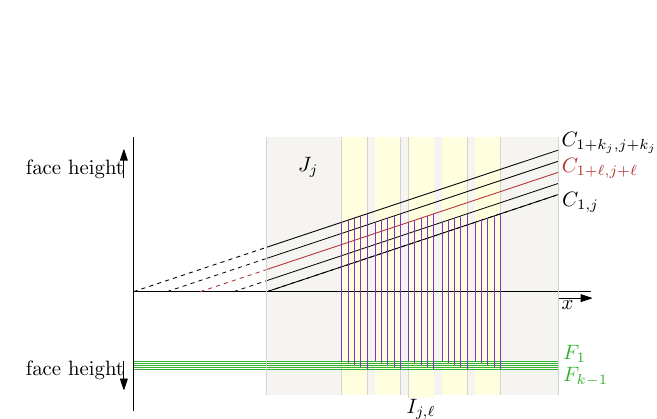}
    \caption{The heights of the faces in a column $\C_j =
      \{C_{1,j},C_{2,j+1}..,C_{1+k_j,j+k_j}\}$ in interval $J_j$.}
    \label{fig:lowerbound_cell_heights}
  \end{figure}

  Consider the (maximal) interval $J_j$ such that all heights of the faces
  (grid cells) in column $j$ are simple increasing linear functions. See
  Fig.~\ref{fig:lowerbound_cell_heights}. It now follows that for each face
  $C_{1+\ell,j+\ell}$ in column $j$ there is a small interval $I_{j,\ell}
  \subset J_j$ in which the height of the face varies between $h$ and
  $h+k\delta'$. Hence, in this interval, the height of face $C_{1+\ell,j+\ell}$
  subsequently becomes equal to the height of the faces $F_1\codots F_k$. We can
  choose $\delta'$ small enough such that the intervals $I_{1+\ell,j+\ell}$ for
  all faces in column $j$ are disjoint (and so that $I_{1+\ell,j+\ell}$ lies to
  the left of $I_{2+\ell,j+\ell+1}$). So, since we have $\Omega(k)$ columns,
  each of $\Omega(k)$ faces, it follows that the number of $x$-coordinates at
  which two faces have equal height is $\Omega(k^3)$.

  Finally, observe that any column $j$, and any $x$-coordinate in $J_j$, the
  faces in column $j$, ordered from top to bottom, have decreasing
  height. Similarly, the faces $F_1 \codots F_k$, ordered from top to bottom, have
  increasing height. It follows that when faces $C_{1+\ell,j+\ell}$ and $F_i$
  have equal height $h'$, all other faces from $\C_\ell$ and $F_1 \codots F_{i-1}$
  have height smaller than $h'$. To make sure the remaining faces (such as the
  triangular face bounded by $r_1$, $b_1$, and $g_1$, have height at most $h$,
  we add a set of grey lines $E$. Since the slope of the blue lines is close to
  one, the distance between $b_1$ and $g_1$, at $x_2$ is at most $O(k)$. The
  same holds for the distance between $r_1$ and $g_1$ at $x_1$. It follows that
  we have to add at most $O(k)$ in between the grid and $g_1$ to make sure all
  intermediate faces have height at most $h$. Hence, all equal-height events
  involving faces $C_{1+\ell,j+\ell}$ and $F_i$ produce a triplet in $S(\A)$.
\end{proof}

\section{A Lower Bound on the Number of Maximal groups}
\label{sec:lower_bound}

In this section we show that, even for fixed parameters $\eps$, $m$, and
$\delta$, the number of maximal $(m,\eps,\delta)$-groups may be as large as
$\Omega(\tau n^3)$ even in $\R$. This strengthens the result of
Buchin~\etal~\cite{grouping2015}, who establish this bound for entities moving
in $\R^2$.

\setcounter{tmpthm}{\value{theorem}}
\setcounterref{theorem}{lem:lowerbound_fixed_eps}
\addtocounter{theorem}{-1}

\begin{lemma}
  For a set \X of $n$ entities, in which each entity travels along a
  piecewise-linear trajectory of $\tau$ edges in $\R^1$, there can be
  $\Omega(\tau n^3)$ maximal $\eps$-groups.
\end{lemma}

\setcounter{theorem}{\value{tmpthm}}

\begin{proof}
  We build a construction in which all entities move along lines that yields
  $\Omega(n^3)$ groups. Repeating this construction $\Omega(\tau)$ times
  produces the claimed bound.

  Partition \X into three sets $R$, $B$, and $H$, with $|R|=|B|=k = \Omega(n)$
  and $|H| = k/2$. The lines (forming the trajectories of the entities) in $R$
  and $B$ form a grid, that we rotate by $45+\delta$ degrees, for some small
  $\delta > 0$. Let $R = \{r_1 \codots r_k\}$, $B = \{b_1 \codots b_k\}$,
  denote the lines in increasing order. See
  Fig.~\ref{fig:lowerbound_fixed_eps}. We partition vertices of the rotated
  grid (i.e. the intersection points of $R$ and $B$) into ``columns'' $C_1
  \codots C_{\Omega(n)}$, where $C_i$ is of the form $\{ (r_j,b_k),
  (r_{j+1},b_{k-1}) \codots \}$, for some $j$ and $k$. For a given column
  $C_i$, and a given line $r_j \in R$, let $b_{ij} \in B$ be the line that
  intersects $r_j$ in $C_i$, and let $t_{ij}$ be the time of the intersection.

  Finally, for each entity $r_j$, $j$ odd, we place a line $h \in H$ through
  the points $r_j \cap b_1$ and $r_{j+1} \cap b_2$. So, in every column $C_i$,
  $h$ intersects some $r_\ell$ in $r_\ell \cap b_{i\ell}$\footnote{Note that we
    can easily perturb the lines to avoid parallel lines and points in which
    three lines intersect.}. We scale the entire construction such that the
  distance between two consecutive lines in $h$ is larger than $\eps$, and
  the distance between an intersection point $r_i \cap b_{ij}$ in column $i$
  and the lines $h_{ij}^-$ and $h_{ij}^+$ in $H$ directly above and below it is
  at most $\eps$.

  \begin{figure}[tb]
    \centering
    \includegraphics{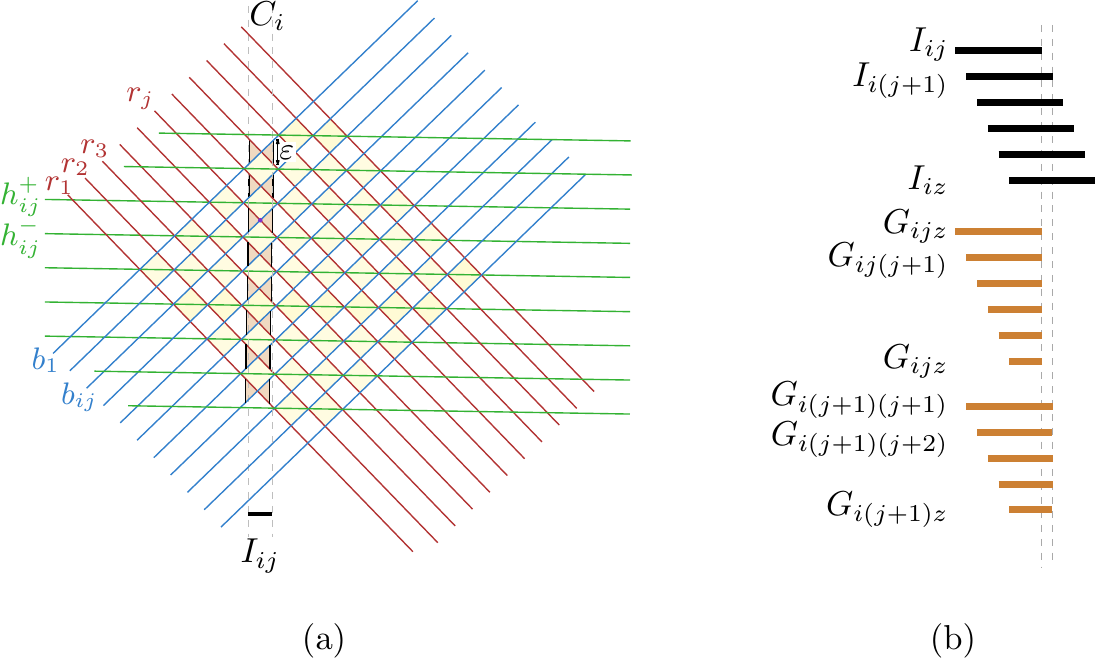}
    \caption{(a) The construction in which the entities move along lines that
      yields $\Omega(n^3)$ maximal $\eps$-groups (for a fixed $\eps$). The
      background color indicates that entities are connected. (b) The intervals
      $I_{ij}$ during which $r_j$ is directly connected to $b_{ij}$ (black), and the
      intervals during which the sets $G_{ijz}$ are maximal (orange).}
    \label{fig:lowerbound_fixed_eps}
  \end{figure}

  Consider a column $C_i$, with $i$ even. Entity $r_j$ is directly connected to
  $b_{ij}$ during some time interval $I_{ij} = [t_{ij} - \Delta, t_{ij} +
  \Delta]$, for some $\Delta$. Since we rotated the grid by $45+\delta$
  degrees, we have $t_{ij} < t_{i(j+1)}$ for every $j$. Hence, every interval
  $I_{i(j+1)}$ starts and ends slightly later than $I_{ij}$. Note that at any
  time during $I_{ij}$, either $r_j$ or $b_{ij}$ is directly connected to
  $h_{ij}^-$. The same holds for $h_{ij}^+$. This means that for any
  consecutive range $[j..z]$, the set of entities $X_{ijz} = \bigcup_{\ell \in
    [j..z]} \{r_\ell,b_{i\ell},h_{i\ell}^-,h_{i\ell}^+\}$ is $\eps$-connected
  during $I_{ijz} = I_{ij} \cap I_{iz} = [t_{iz}-\Delta,t_{ij}+\Delta]$. It is
  easy to see that $I_{ijz}$ is maximal in duration.

  It now follows that a column $C_i$ with $\Omega(n)$ intersection points
  ``generates'' $\Omega(n^2)$ maximal $\eps$-groups. Let $r_j \codots r_z$ be
  entities from $R$ involved in column $C_i$. At time $t_{i\ell} - \Delta$, for
  $\ell \in [j..z]$, the group $G_{ij\ell}$ starts as a new maximal group. When
  $r_j$ and $b_{ij}$ disconnect, all these groups end, and the groups
  $G_{i(j+1)(j+1)} \codots G_{i(j+1)z}$ are discovered as new maximal
  $\eps$-groups. Since we have $\Omega(n)$ columns (that have $\Omega(n)$
  intersection points) we get $\Omega(n^3)$ maximal $\eps$-groups as
  desired.
\end{proof}

Note that we can choose the speed of the entities such that this construction
holds for any choice of $\delta$. Furthermore, the construction holds even for
minimal group sizes of $\Omega(n)$.

\end{document}